\documentclass[a4paper,USenglish,cleveref,autoref]{lipics-v2019}
\usepackage{microtype}
\usepackage{wrapfig}
\usepackage{graphicx}
\usepackage{xspace}
\usepackage{multicol}
\usepackage{cite}
\usepackage[table]{xcolor}
\usepackage[ruled,vlined,linesnumbered]{algorithm2e}
\graphicspath{{img/}}
\newcommand{\RE}{\mathbb{R}}

\newcommand{\bigOh}{\mathcal{O}}

\newcommand{\OO}[1]{O\kern-2pt\left(#1\right)}  

\newcommand{\setP}{\ensuremath{P}\xspace}
\newcommand{\setR}{\ensuremath{R}\xspace}

\newcommand{\ie}{\emph{i.e.},\xspace}
\newcommand{\eg}{\emph{e.g.},\xspace}
\newcommand{\wrt}{\emph{w.r.t.}\xspace}

\newcommand{\numNE}{\kappa}

\newcommand{\spread}{\Delta}
\newcommand{\eps}{\varepsilon}

\newcommand{\metricSet}{\mathcal{X}}
\newcommand{\metricFunc}{\textup{\textsf{d}}}
\newcommand{\metricSpace}{$(\metricSet,\metricFunc)$\xspace}
\newcommand{\ddim}{\textup{ddim}(\metricSet)}

\newcommand{\dist}[2]{\metricFunc(#1,#2)}
\newcommand{\nn}[1]{\textup{nn}(#1)}
\newcommand{\nenemy}[1]{\textup{ne}(#1)}
\newcommand{\dnn}[1]{\metricFunc_\textup{nn}(#1)}
\newcommand{\dne}[1]{\metricFunc_\textup{ne}(#1)}
\newcommand{\chromdens}[1]{\delta(#1)}
\newcommand{\NN}{\textup{NN}\xspace}
\newcommand{\NE}{\textup{NE}\xspace}
\newcommand{\ANN}{\textup{ANN}\xspace}
\newcommand{\SNN}{\textup{SNN}\xspace}
\newcommand{\CNN}{\textup{CNN}\xspace}
\newcommand{\FCNN}{\textup{FCNN}\xspace}
\newcommand{\alphaFCNN}{\mbox{\textup{$\alpha$-FCNN}}\xspace}
\newcommand{\SFCNN}{\textup{SFCNN}\xspace}
\newcommand{\alphaSFCNN}{\mbox{\textup{$\alpha$-SFCNN}}\xspace}

\newcommand{\NET}{\textup{NET}\xspace}
\newcommand{\alphaNET}{\mbox{\textup{$\alpha$-NET}}\xspace}
\newcommand{\MSS}{\textup{MSS}\xspace}
\newcommand{\VSS}{\textup{VSS}\xspace}
\newcommand{\RSS}{\textup{RSS}\xspace}
\newcommand{\alphaRSS}{\mbox{\textup{$\alpha$-RSS}}\xspace}
\newcommand{\paramRSS}[1]{\mbox{\textup{$#1$-RSS}}\xspace}

\newcommand{\alphaHSS}{\mbox{\textup{$\alpha$-HSS}}\xspace}

\definecolor{yellowcd}{RGB}{252, 229, 30}
\definecolor{bluecd}{RGB}{51, 0, 68}
\newcommand{\careful}[1]{\textcolor{red}{#1}}
\newcommand{\alejandro}[1]{{\textcolor{red}{[\textbf{Alejandro}: {#1}]}}}

\title{Coresets for the Nearest-Neighbor Rule}

\author{Alejandro Flores-Velazco}{
    Department of Computer Science\\
    University of Maryland, College Park, MD, USA}
    {afloresv@cs.umd.edu}
    {https://orcid.org/0000-0003-0868-9802}{}
\author{David M. Mount}{
    Department of Computer Science and Institute for Advanced Computer Studies\\
    University of Maryland, College Park, MD, USA}
    {mount@umd.edu}
    {http://orcid.org/0000-0002-3290-8932}{}
\authorrunning{A.\ Flores\,-Velazco and D.\ Mount}
\Copyright{A.\ Flores\,-Velazco and D.\ Mount}

\category{}
\relatedversion{A preliminary version of this paper appeared in ESA~2020~\cite{esa20afloresv}.}
\supplement{Source code is available at \url{https://github.com/afloresv/nnc}}
\funding{Research partially supported by NSF grant CCF-1618866.}
\acknowledgements{Thanks to Prof. Emely Arr\'aiz for suggesting the problem of condensation while the first author was a student at Universidad Sim\'on Bol\'ivar, Venezuela. Thanks to Ahmed~Abdelkader for the helpful discussions and valuable suggestions.}
\nolinenumbers
\hideLIPIcs
\EventEditors{}
\EventNoEds{0}
\EventLongTitle{}
\EventShortTitle{}
\EventAcronym{}
\EventYear{}
\EventDate{}
\EventLocation{}
\EventLogo{}
\SeriesVolume{}
\ArticleNo{}
\begin{document}

\maketitle

\begin{abstract}
Given a training set \setP of \emph{labeled} points, the \emph{nearest-neighbor rule} predicts the class of an \emph{unlabeled} query point as the label of its closest point in the set. To improve the time and space complexity of classification, a natural question is how to reduce the training set without significantly affecting the accuracy of the nearest-neighbor rule. \emph{Nearest-neighbor condensation} deals with finding a subset $\setR \subseteq \setP$ such that for every point $p \in \setP$, $p$'s nearest-neighbor in \setR has the same label as $p$. This relates to the concept of \emph{coresets}, which can be broadly defined as subsets of the set, such that an exact result on the coreset corresponds to an approximate result on the original set. However, the guarantees of a coreset hold for any query point, and not only for the points of the training set.

This paper introduces the concept of coresets for nearest-neighbor classification. We extend existing criteria used for condensation, and prove sufficient conditions to correctly classify any query point when using these subsets. Additionally, we prove that finding such subsets of minimum cardinality is NP-hard, and propose quadratic-time approximation algorithms with provable upper-bounds on the size of their selected subsets. Moreover, we show how to improve one of these algorithms to have subquadratic runtime, being the first of this kind for condensation.
\end{abstract}

\section{Introduction}

In non-parametric classification, we are given a \emph{training set} \setP consisting of $n$ points in a metric space \metricSpace, with domain $\metricSet$ and distance function $\metricFunc: \metricSet^2 \rightarrow \RE^{+}$.
Additionally, \setP is partitioned into a finite set of \emph{classes} by associating each point $p \in \setP$ with a \emph{label} $l(p)$, indicating the class to which it belongs. Given an \emph{unlabeled} query point $q \in \metricSet$, the goal of a \emph{classifier} is to predict $q$'s label using the training set \setP.

The \emph{nearest-neighbor rule} is among the best-known classification techniques~\cite{fix_51_discriminatory}. It assigns a query point the label of its closest point in $\setP$, according to the metric~$\metricFunc$.
The nearest-neighbor rule exhibits good classification accuracy both experimentally and theoretically \cite{stone1977consistent,Cover:2006:NNP:2263261.2267456,devroye1981inequality}, but
it is often criticized due to its high space and time complexities. Clearly, the training set \setP must be stored to answer nearest-neighbor queries, and the time required for such queries depends to a large degree on the size and dimensionality of the data.
These drawbacks inspire the question of whether it is possible replace \setP with a significantly smaller~subset, without significantly reducing the classification accuracy under the nearest-neighbor rule. This problem~is called 
 \emph{nearest-neighbor~condensation}~\cite{Hart:2006:CNN:2263267.2267647,ritter1975algorithm,gottlieb2014near,DBLP:conf/jcdcg/Toussaint02}.

There are obvious parallels between condensation and the concept of \emph{coresets} in geometric approximation~\cite{agarwal2005geometric,phillips2016coresets,feldman2020core,har2004coresets}. Intuitively, a \emph{coreset} is small subset of the original data,~that well approximates some statistical properties of the original set.
Coresets have also been applied to many problems in machine learning, such as clustering and neural network compression~\cite{baykal2018data,braverman2016new,feldman2011unified,liebenwein2019provable}. This includes recent results on coresets for the SVM classifier~\cite{tukan2020coresets}.

This paper presents the first approach to compute coresets for the nearest-neighbor rule, leveraging its resemblance to the problem of nearest-neighbor condensation. We also present one of the first results on practical condensation algorithms with theoretical guarantees.

%
%

\subparagraph*{Preliminaries.}
Given any point $q \in \metricSet$ in the metric space, its nearest-neighbor, denoted $\nn{q}$, is the closest point of \setP according the the distance function $\metricFunc$. The distance from $q$ to its nearest-neighbor is denoted by $\dnn{q,\setP}$, or simply $\dnn{q}$ when \setP is clear. Given a point $p \in \setP$ from the training set, its nearest-neighbor in \setP is point $p$ itself. Additionally, any point of $\setP$ whose label differs from $p$'s is called an \emph{enemy} of $p$. The closest such point is called $p$'s \emph{nearest-enemy}, and the distance to this point is called $p$'s \emph{nearest-enemy distance}. These are denoted by $\nenemy{p}$ and $\dne{p,\setP}$ (or simply $\dne{p}$), respectively. 

Clearly, the size of a coreset for nearest-neighbor classification depends on the spatial characteristics of the classes in the training set. For example, it is much easier to find a small coreset for two spatially well separated clusters than for two classes that have a high degree of overlap. To model the intrinsic complexity of nearest-neighbor classification, we define $\numNE$ to be the number of nearest-enemy points of \setP, \ie the cardinality of set $\{\nenemy{p} \mid p \in \setP\}$.

%
%

Through a suitable uniform scaling, we may assume that the \emph{diameter} of \setP (that is, the maximum distance between any two points in the training set) is 1. The \emph{spread} of \setP, denoted as $\spread$, is the ratio between the largest and smallest distances in \setP. Define the \emph{margin} of \setP, denoted $\gamma$, to be the smallest nearest-enemy distance in \setP. Clearly, $1/\gamma \leq \spread$.

A metric space \metricSpace is said to be \emph{doubling}~\cite{heinonen2012lectures} if there exist some bounded value $\lambda$ such~that any metric ball of radius $r$ can be covered with at most $\lambda$ metric balls of radius $r/2$. Its \emph{doubling dimension} is the base-2 logarithm of $\lambda$, denoted as $\ddim = \log{\lambda}$. Throughout, we assume that $\ddim$ is a constant, which means that multiplicative factors depending on $\ddim$ may be hidden in our asymptotic notation. Many natural metric spaces of interest are doubling, including $d$-dimensional Euclidean space whose doubling dimension is $\Theta(d)$. It is well know that for any subset $R \subseteq \metricSet$ with some spread $\spread_R$, the size of $R$ is bounded by $|R| \leq \lceil\spread_R\rceil^{\ddim+1}$.

\subparagraph*{Related Work.}
A subset $\setR \subseteq \setP$ is said to be \emph{consistent}~\cite{Hart:2006:CNN:2263267.2267647} if and only if for every $p \in \setP$~its nearest-neighbor in \setR is of the same class as $p$. Intuitively, \setR is consistent if and only if all points of \setP are correctly classified using the nearest-neighbor rule over \setR. Formally, the problem of \emph{nearest-neighbor condensation} consists of finding a consistent subset of \setP.

Another criterion used for condensation is known as \emph{selectiveness}~\cite{ritter1975algorithm}. A subset $\setR \subseteq \setP$~is said to be \emph{selective} if and only if for all $p \in \setP$ its nearest-neighbor in \setR is closer to~$p$~than its nearest-enemy in \setP. Clearly, any selective subset is also consistent.
Observe that these condensation criteria ensure that every point in the training set will be correctly classified after condensation, but they do not imply the same for arbitrary points in the metric space.

\begin{figure*}[t]
    \centering
    \begin{subfigure}[b]{.25\linewidth}
        \centering\includegraphics[width=.9\textwidth]{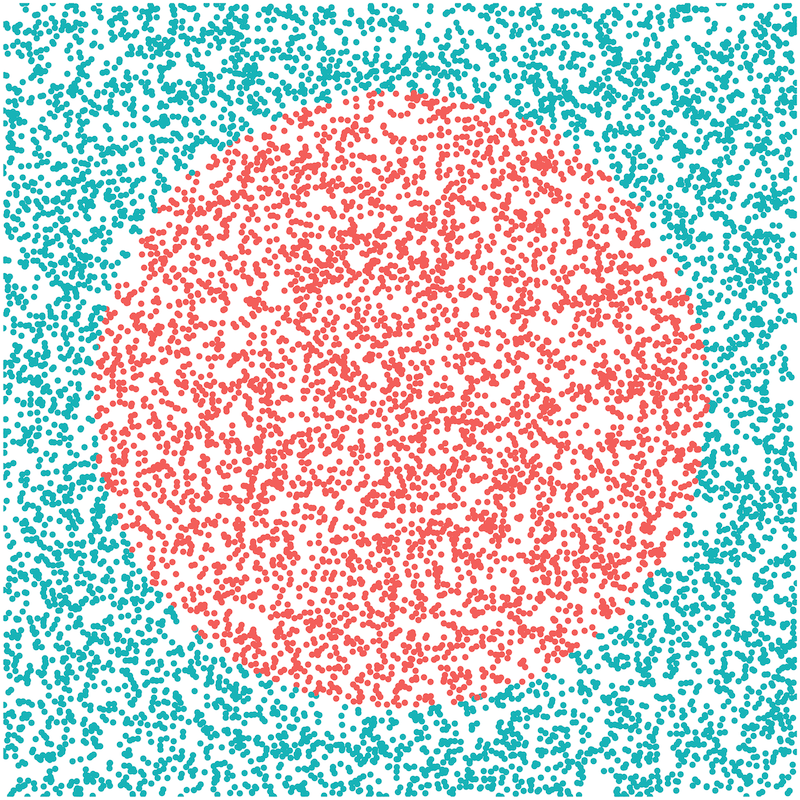}
        \caption{Training set ($10^4$\,pts)}\label{fig:algexample:dataset}
    \end{subfigure}%
    \newcommand{\printalgexample}[3]{%
        \begin{subfigure}[b]{.25\linewidth}
            \centering\includegraphics[width=.9\textwidth]{sel/#1.pdf}
            \caption{#3}\label{fig:algexample:#2}
        \end{subfigure}%
    }%
    \printalgexample{FCNN}{fcnn}{\FCNN (222 pts)}%
    \printalgexample{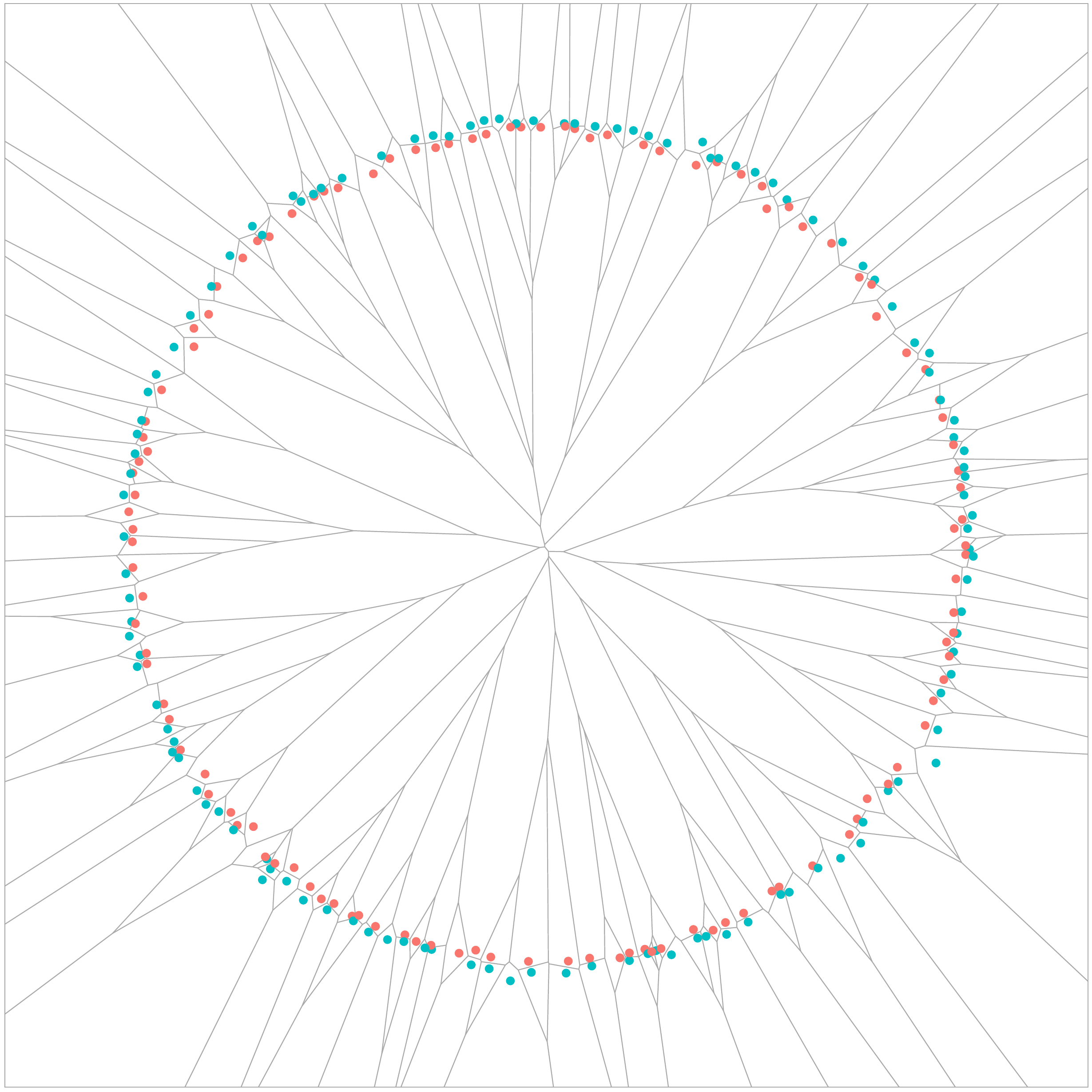}{vss}{\VSS (233 pts)}%
    \printalgexample{RSS}{rss}{\RSS (233 pts)}%
    
    \bigskip
    \printalgexample{01-RSS}{0.1-rss}{\paramRSS{0.1} (300 pts)}%
    \printalgexample{05-RSS}{0.5-rss}{\paramRSS{0.5} (540 pts)}%
    \printalgexample{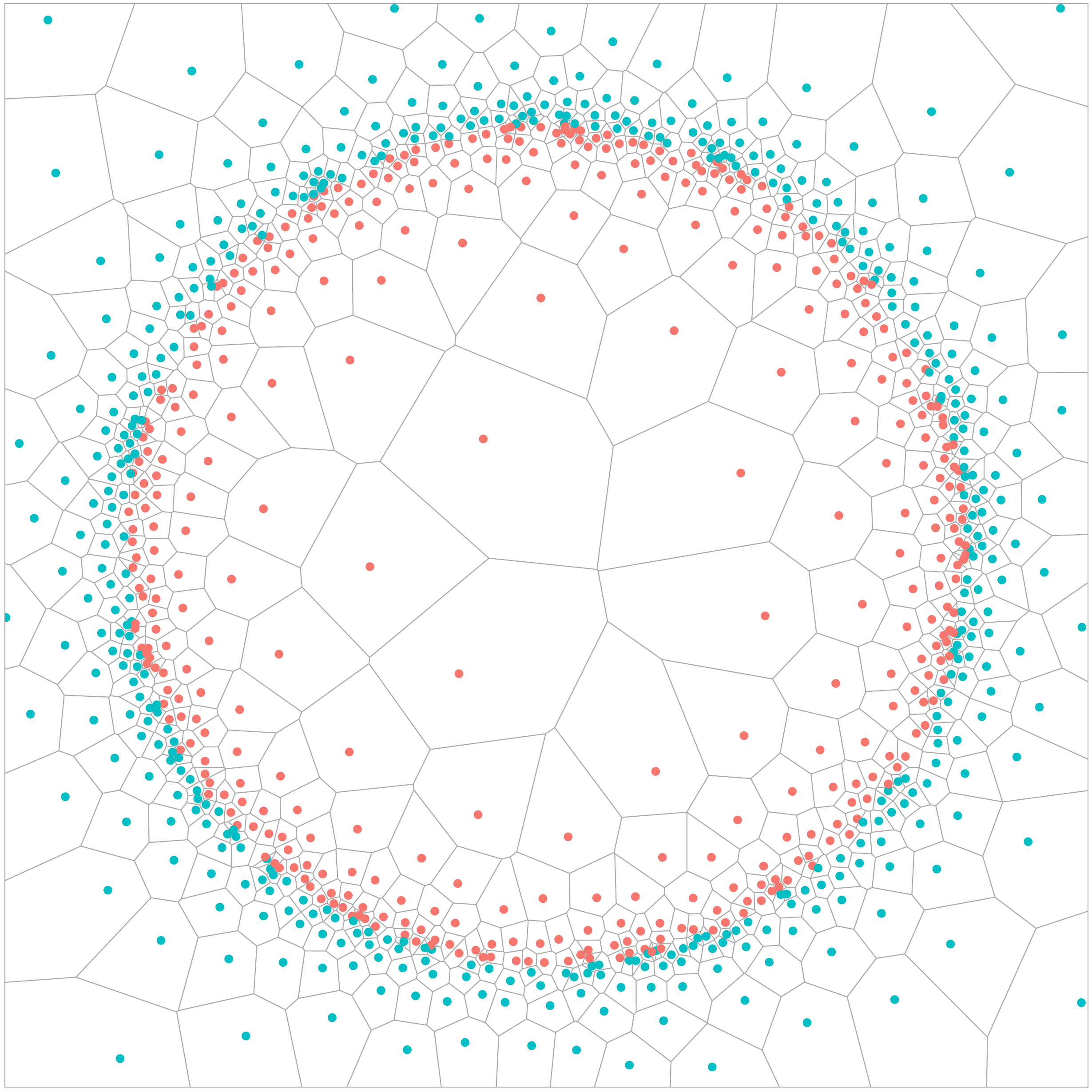}{1-rss}{\paramRSS{1} (846 pts)}%
    \printalgexample{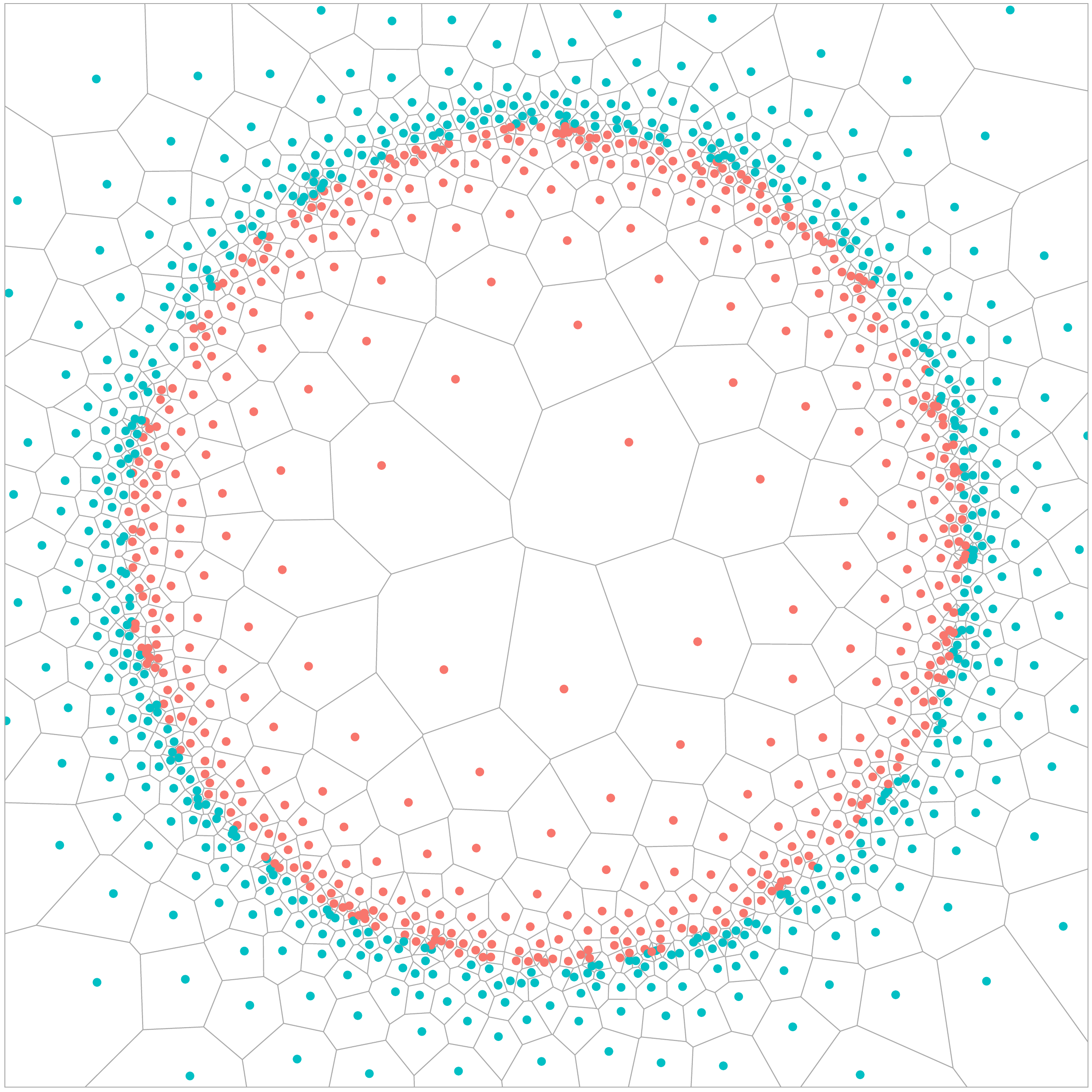}{1.4142-rss}{\paramRSS{\sqrt{2}} (1066 pts)}%
    
    \caption{An illustrative example of the subsets selected by different condensation algorithms from an initial training set \setP in $\RE^2$ of $10^4$ points. \FCNN, \VSS, and \RSS, are known algorithms for this problem, while \alphaRSS is proposed in this paper, along with new condensation criteria. The subsets selected by \alphaRSS depend on the parameter $\alpha \geq 0$, here assigned to the values $\alpha = \{0.1,0.5,1,\sqrt{2}\}$.}\label{fig:algexample} 
    \vspace*{-10pt}
\end{figure*}

%
%

It is known that the problems of computing consistent and selective subsets of minimum cardinality are both NP-hard~\cite{Wilfong:1991:NNP:109648.109673,Zukhba:2010:NPP:1921730.1921735,khodamoradi2018consistent}.
An approximation algorithm called \NET~\cite{gottlieb2014near} was proposed for the problem of finding minimum cardinality consistent subsets, along with almost matching hardness lower-bounds. The algorithm simply computes a $\gamma$-net of \setP, where $\gamma$ is the minimum nearest-enemy distance in \setP, which clearly results in a consistent subset of \setP (also selective). In practice, $\gamma$ tends to be small, which results in subsets of much higher cardinality than needed. To overcome this issue, the authors proposed a post-processing pruning technique to further reduce the selected subset. Even with the extra pruning, \NET is often outperformed on typical data sets by more practical heuristics with respect to runtime and selection size. More recently, a subexponential-time algorithm was proposed~\cite{biniaz2019minimum} for finding minimum cardinality consistent subsets of point sets $\setP \subset \RE^2$ in the plane, along with other case-specific algorithms for special instances of the problem in $\RE^2$.
On the other hand, less is known about computing minimum cardinality selective subsets: there is only a worst-case exponential time algorithm called \SNN~\cite{ritter1975algorithm} for computing such optimal subsets.

Most research has focused on proposing practical heuristics to find either consistent or selective subsets of \setP (for comprehensive surveys see \cite{DBLP:conf/jcdcg/Toussaint02,jankowski2004comparison}). 
\CNN~(\emph{Condensed Nearest-Neighbor})~\cite{Hart:2006:CNN:2263267.2267647} was the first algorithm proposed to compute consistent subsets. Even though it has been widely used in the literature, \CNN suffers from several drawbacks: its running time is cubic in the worst-case, and the resulting subset is \emph{order-dependent}, meaning that the result is determined by the order in which points are considered by the algorithm.
Alternatives include \FCNN (\emph{Fast} \CNN)~\cite{angiulli2007fast} and \MSS (\emph{Modified Selective Subset})~\cite{barandela2005decision}, which compute consistent and selective subsets respectively. Both algorithms run in $\bigOh(n^2)$ worst-case time, and are order-independent.
While such heuristics have been extensively studied experimentally~\cite{Garcia:2012:PSN:2122272.2122582}, theoretical results are scarce. Recently, we have shown~\cite{DBLP:conf/cccg/Flores-VelazcoM19, esa20afloresv} that~the size of the subsets selected by \MSS and \FCNN cannot be bounded. Alternatively, these papers propose three new quadratic-time algorithms that are both efficient in practice, and have provable upper-bounds on their selection size. These algorithms are called \RSS~(\emph{Relaxed Selective Subset}) and \VSS~(\emph{Voronoi Selective Subset}) for finding selective subsets, and \SFCNN (\emph{Single} \FCNN) for finding consistent subsets.

\subparagraph*{Contributions.}

As mentioned in the previous section, consistency and selectivity imply correct classification to points of the training set, but not to arbitrary points of the metric space (This is striking since this is the fundamental purpose of classification!). In this paper, we introduce the concept of a coreset for classification with the nearest-neighbor rule, which provides approximate guarantees on correct classification for all query points. We demonstrate their existence, analyze their size, and discuss their efficient computation. 

We say that a subset $\setR \subseteq \setP$ is an \emph{$\eps$-coreset for the nearest-neighbor rule} on \setP, if and only if for every query point $q \in \metricSet$, the class of its exact nearest-neighbor in \setR is the same as the class of some $\eps$-approximate nearest-neighbor of $q$ in \setP (see Section~\ref{sec:approx-sensitive-nnc} for definitions). Recalling the concepts of $\numNE$ and $\gamma$ introduced in the preliminaries, here is our main result:

\begin{theorem}
\label{thm:coreset:main}
Given a training set \setP in a doubling metric space \metricSpace, there exist an $\eps$-coreset for the nearest-neighbor rule of size $\bigOh(\numNE\,\log{\frac{1}{\gamma}}\,(1/\eps)^{\ddim+1})$, and this coreset can be computed in subquadratic worst-case time.
\end{theorem}

\noindent
Here is a summary of our principal results:
\begin{itemize}
\item We extend the criteria used for nearest-neighbor condensation, and identify sufficient conditions to guarantee the correct classification of any query point after condensation. 

\item We prove that finding minimum-cardinality subsets with this new criteria is NP-hard.

\item We provide quadratic-time approximation algorithms with provable upper-bounds on the sizes of their selected subsets, and we show that the running time of one such algorithm can be improved to be subquadratic.
\end{itemize}

Our subquadratic-time algorithm is the first with such worst-case runtime for the problem of nearest-neighbor condensation.

\section{Coreset Characterization}
\label{sec:approx-sensitive-nnc}

In practice, nearest-neighbors are usually not computed exactly, but rather approximately. Given an approximation parameter $\eps \geq 0$, an $\eps$-\emph{approximate} nearest-neighbor or $\eps$-\ANN query returns any point whose distance from the query point is within a factor of $(1+\eps)$ times the true nearest-neighbor distance. 

Intuitively, a query point should be easier to classify if its nearest-neighbor is significantly closer than its nearest-enemy.
This intuition can be formalized through the concept of the \emph{chromatic density}~\cite{MOUNT200097} of a query point $q \in \metricSet$ with respect to a set $\setR \subseteq \setP$, defined as:
\begin{equation}
    \chromdens{q,\setR} = \frac{\dne{q,\setR}}{\dnn{q,\setR}} -1.
\end{equation}

Clearly, if $\chromdens{q,\setR} > \eps$ then $q$ will be correctly classified%
\footnote{By \emph{correct classification}, we mean that the classification is the same as the classification that results from applying the nearest-neighbor rule exactly on the entire training set \setP.}
by an $\eps$-\ANN query \mbox{over \setR, as all} possible candidates for the approximate nearest-neighbor belong to the same class as $q$'s true nearest-neighbor. However, as evidenced in Figures~\ref{fig:cdheatmap:fcnn} and~\ref{fig:cdheatmap:rss}, one side effect of existing condensation algorithms is a significant reduction in the chromatic density of query points. Consequently, we propose new criteria and algorithms that maintain high chromatic densities after condensation, which are then leveraged to build coresets for the nearest-neighbor rule.

\subsection{Approximation-Sensitive Condensation}

The decision boundaries of the nearest-neighbor rule (that is, points $q$ such that $\dne{q,\setP} = \dnn{q,\setP}$) are naturally characterized by points that separate clusters of points of different classes. As illustrated in Figures~\ref{fig:algexample:fcnn}-\ref{fig:algexample:rss}, condensation algorithms tend to select such points. However, this behavior implies a significant reduction of the chromatic density of query points that are far from such boundaries (see Figures~\ref{fig:cdheatmap:fcnn}-\ref{fig:cdheatmap:rss}).

\begin{figure*}[t!]
    \centering
    \begin{subfigure}[b]{.25\linewidth}
        \centering\includegraphics[width=.9\textwidth]{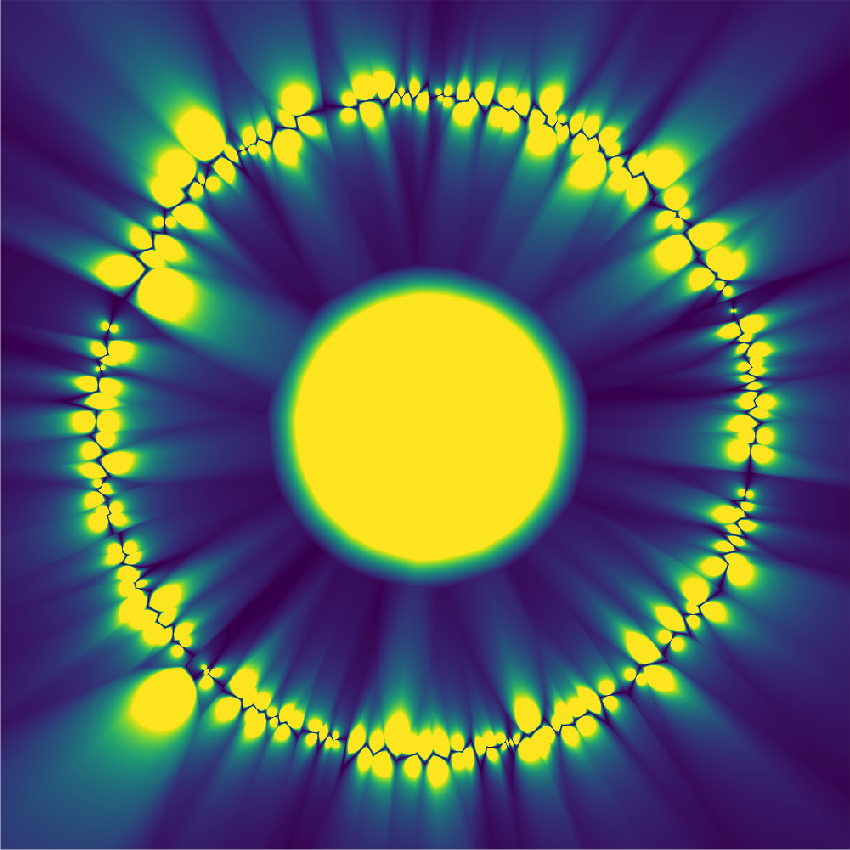}
        \caption{\FCNN}\label{fig:cdheatmap:fcnn}
    \end{subfigure}%
    \begin{subfigure}[b]{.25\linewidth}
        \centering\includegraphics[width=.9\textwidth]{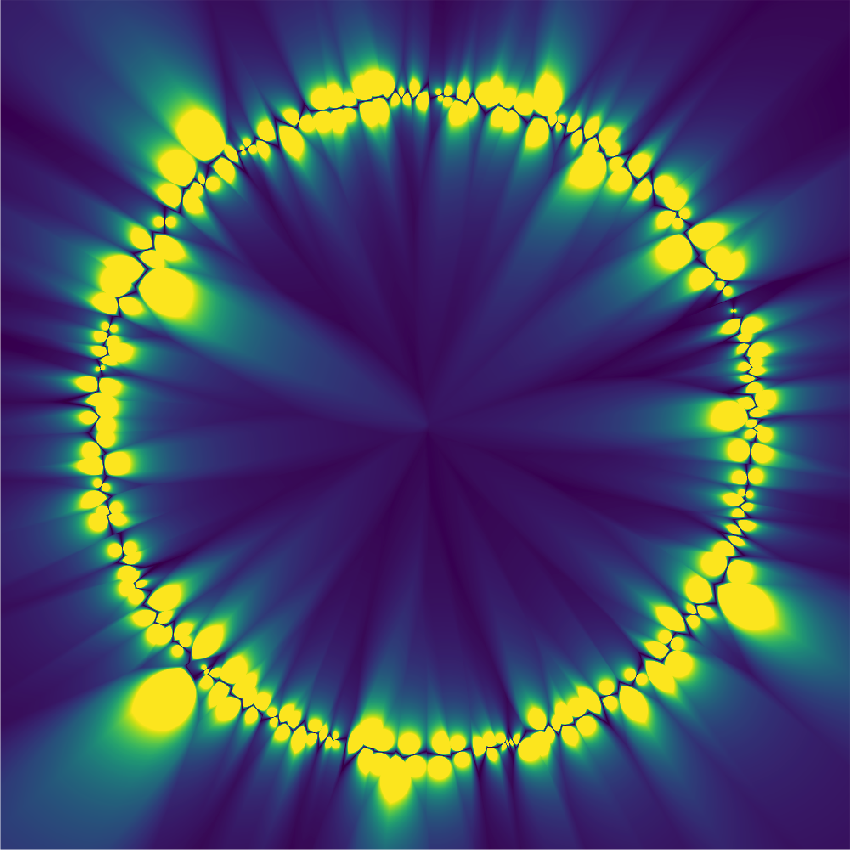}
        \caption{\RSS}\label{fig:cdheatmap:rss}
    \end{subfigure}%
    \begin{subfigure}[b]{.25\linewidth}
        \centering\includegraphics[width=.9\textwidth]{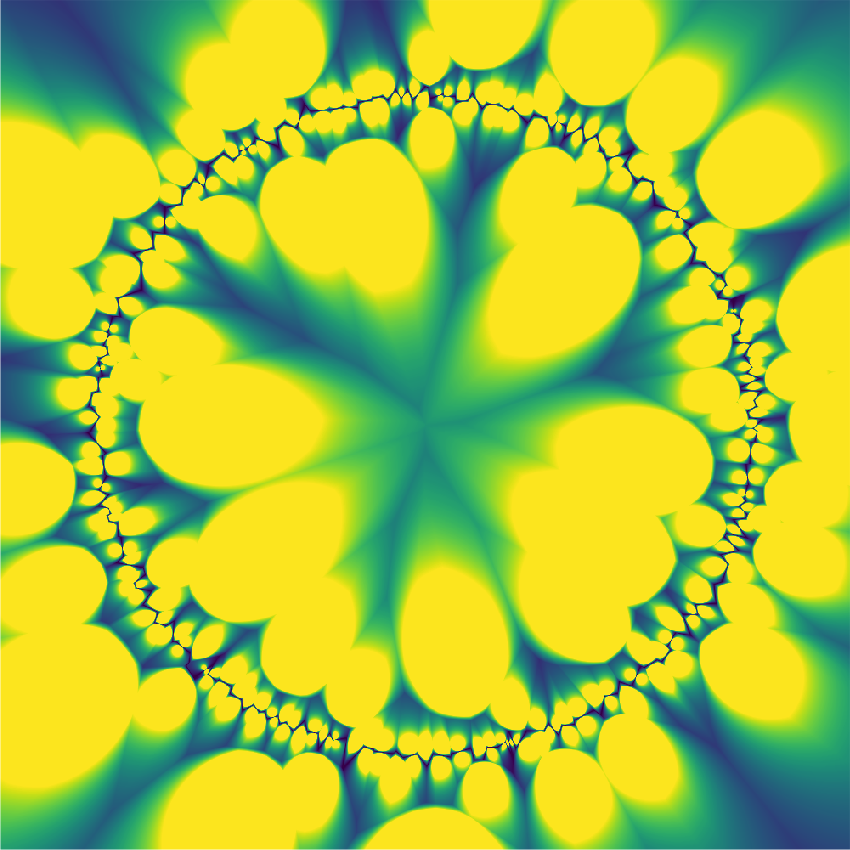}
        \caption{\paramRSS{0.1}}\label{fig:cdheatmap:arss-0.1}
    \end{subfigure}%
    \begin{subfigure}[b]{.25\linewidth}
        \centering\includegraphics[width=.9\textwidth]{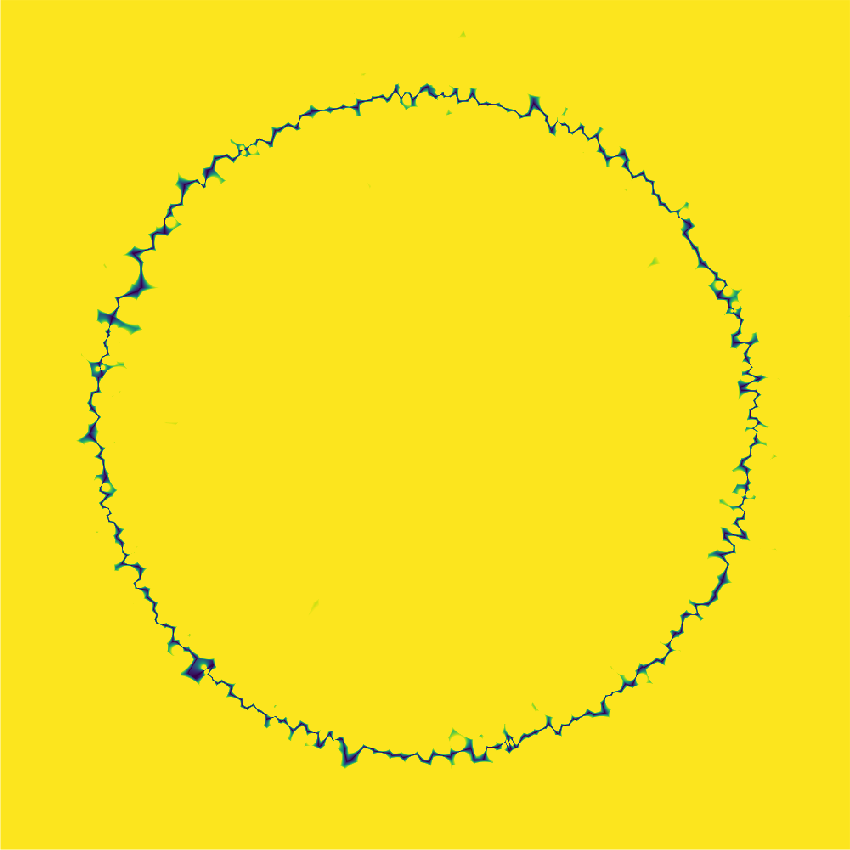}
        \caption{\paramRSS{0.5}}\label{fig:cdheatmap:arss-0.5}
    \end{subfigure}%
    \caption{Heatmap of \emph{chromatic density} values of points in $\RE^2$ \wrt the subsets computed by different condensation algorithms: \FCNN, \RSS, and \alphaRSS (see Figure~\ref{fig:algexample}). \emph{Yellow}~{\color{yellowcd}$\bullet$} corresponds to chromatic density values $\geq 0.5$, while \emph{blue}~{\color{bluecd}$\bullet$} corresponds to $0$. Evidently, \alphaRSS helps maintaining high chromatic density values when compared to standard condensation algorithms.}\label{fig:cdheatmap}
\end{figure*}

A natural way to define an approximate notion of consistency is to ensure that all points in \setP are correctly classified by \ANN queries over the condensed subset \setR.
%
Given a condensation parameter $\alpha \geq 0$, we define a subset $\setR \subseteq \setP$ to be:
\vspace*{5pt}
\begin{description}
    \setlength{\itemsep}{5pt}
    \item[$\alpha$-consistent] if
        $\forall\ p \in \setP,\ \dnn{p,\setR} < \dne{p,\setR}/(1+\alpha)$.
    \item[$\alpha$-selective]\ \ if 
        $\forall\ p \in \setP,\ \dnn{p,\setR} < \dne{p,\setP}/(1+\alpha)$.
\end{description}

It is easy to see that the standard forms arise as special cases when $\alpha = 0$. These new condensation criteria imply that $\chromdens{p,\setR} > \alpha$ for every $p \in \setP$, meaning that $p$ is correctly classified using an $\alpha$-\ANN query on \setR.
Note that any $\alpha$-selective subset is also $\alpha$-consistent.
Such subsets always exist for any $\alpha \geq 0$ by taking $\setR = \setP$. Current condensation algorithms cannot guarantee either $\alpha$-consistency or $\alpha$-selectiveness unless $\alpha$ is equal to zero. Therefore, the central algorithmic challenge is how to efficiently compute such sets whose sizes are significantly smaller than \setP.
We propose new algorithms to compute such subsets, which showcase how to maintain high chromatic density values after condensation, as evidenced in Figures~\ref{fig:cdheatmap:arss-0.1} and \ref{fig:cdheatmap:arss-0.5}. This empirical evidence is matched with theoretical guarantees for $\alpha$-consistent and $\alpha$-selective subsets, described in the following section.

\subsection{Guarantees on Classification Accuracy}

These newly defined criteria for nearest-neighbor condensation enforce lower-bounds on the chromatic density of any point of \setP after condensation. However, this doesn't immediately imply having similar lower-bounds for unlabeled query points of $\metricSet$. 
In this section, we prove useful bounds on the chromatic density of query points, and characterize sufficient conditions to correctly classify some of these query points after condensation.

Intuitively, the chromatic density determines how easy it is to correctly classify a query point $q \in \metricSet$. We show that the ``ease'' of classification of $q$ after condensation (\ie~$\chromdens{q,\setR}$) depends on both the condensation parameter $\alpha$, and the chromatic density of $q$ before condensation (\ie~$\chromdens{q,\setP}$). This result is formalized in the following lemma:

\begin{lemma}
\label{lemma:bound-chromdens}
Let $q \in \metricSet$ be a query point, and \setR an $\alpha$-consistent subset of \setP, for $\alpha \geq 0$. Then, $q$'s chromatic density with respect to \setR is:
\begin{equation*}
    \chromdens{q,\setR} > \frac{\alpha \, \chromdens{q,\setP} - 2}{\chromdens{q,\setP} + \alpha + 3}.
\end{equation*}
\end{lemma}

\begin{proof}
The proof follows by analyzing $q$'s nearest-enemy distance in \setR. To this end, consider the point $p \in \setP$ that is $q$'s nearest-neighbor in \setP. There are two possible cases:

\begin{description}
    \item[Case 1:] If $p \in \setR$, clearly $\dnn{q,\setR} = \dnn{q,\setP}$. Additionally, it is easy to show that after condensation, $q$'s nearest-enemy distance can only increase: \ie~$\dne{q,\setP} \leq \dne{q,\setR}$. This implies that $\chromdens{q,\setR} \geq \chromdens{q,\setP}$.
    \item[Case 2:] If $p \not\in \setR$, we can upper-bound $q$'s nearest-neighbor distance in \setR as follows:
\end{description}

Since \setR is an $\alpha$-consistent subset of \setP, we know that there exists a point $r \in \setR$~such~that $\dist{p}{r} < \dne{p,\setR}/(1+\alpha)$. By the triangle inequality and the definition of nearest-enemy, $\dne{p,\setR} \leq \dist{p}{\nenemy{q,\setR}} \leq \dist{q}{p} + \dne{q,\setR}$. Additionally, applying the definition of chromatic density on $q$ and knowing that $\dne{q,\setP} \leq \dne{q,\setR}$, we have $\dist{q}{p} = \dnn{q,\setP} \leq \dnn{q,\setR} = \dne{q,\setR}/(1+\chromdens{q,\setP})$. Therefore:
\begin{align*}
    \dnn{q,\setR} 
        \leq \dist{q}{r}
        &\leq \dist{q}{p} + \dist{p}{r}\\
        &< \dist{q}{p} + \frac{\dist{q}{p} + \dne{q,\setR}}{1+\alpha}
         \leq \left( \frac{\chromdens{q,\setP} + \alpha + 3}{(1+\alpha)(1+\chromdens{q,\setP})} \right) \dne{q,\setR}.
\end{align*}
\vspace*{10pt}
Finally, from the definition of $\chromdens{q,\setR}$, we have:\\
\indent\(
\displaystyle
\chromdens{q,\setR}
    = \frac{\dne{q,\setR}}{\dnn{q,\setR}} - 1
    > \frac{(1+\alpha)(1+\chromdens{q,\setP})}{\chromdens{q,\setP}+\alpha+3} - 1
    = \frac{\alpha\,\chromdens{q,\setP} - 2}{\chromdens{q,\setP} + \alpha + 3}.
\)
\end{proof}

The above result can be leveraged to define a coreset, in the sense that an exact result on the coreset corresponds to an approximate result on the original set. As previously defined, we say that a set $\setR \subseteq \setP$ is an \emph{$\eps$-coreset for the nearest-neighbor rule} on \setP, if and only if for every query point $q \in \metricSet$, the class of $q$'s exact nearest-neighbor in \setR is the same as the class of any of its $\eps$-approximate nearest-neighbors in \setP.

\begin{lemma}
\label{thm:coreset:isconsistent}
Any $\eps$-coreset for the nearest-neighbor rule is an $\alpha$-consistent subset, for $\alpha \geq 0$.
\end{lemma}

\begin{proof}
Consider any $\eps$-coreset $C \subseteq \setP$ for the nearest-neighbor rule on \setP. Since the approximation guarantee holds for any point in $\metricSet$, it holds for any $p \in \setP \setminus C$. We know $p$'s nearest-neighbor in the original set \setP is $p$ itself, thus making $\dnn{p,\setP}$ zero. This implies that $p$ must be correctly classified by a nearest-neighbor query on $C$, that is, $\dnn{p,C} < \dne{p,C}$, which is the definition of $\alpha$-consistency for any $\alpha \geq 0$.
\end{proof}

\begin{theorem}
\label{thm:coreset:nn}
Any $2/\eps$-selective subset is an $\eps$-coreset for the nearest-neighbor rule.
\end{theorem}

\begin{proof}
Let \setR be an $\alpha$-selective subset of \setP, where $\alpha = 2/\eps$. Consider any query point $q \in \metricSet$ in the metric space. It suffices to show that its nearest-neighbor in \setR is of the same class as any $\eps$-approximate nearest-neighbor in \setP. To this end, consider $q$'s chromatic density with respect to both \setP and \setR, denoted as $\chromdens{q,\setP}$ and $\chromdens{q,\setR}$, respectively. We identify two cases:

\begin{description}
    \setlength{\itemsep}{5pt}

    \item[Case 1 (Correct-Classification guarantee):]
    If $\chromdens{q,\setP} \geq \eps$.\\
    Consider the bound derived in Lemma~\ref{lemma:bound-chromdens}. Since $\alpha \geq 0$, and by our assumption that $\chromdens{q,\setP} \geq \eps > 0$, setting $\alpha=2/\eps$ implies that $\chromdens{q,\setR} > 0$. This means that the nearest-neighbor of $q$ in \setR belongs to the same class as the nearest-neighbor of $q$ in \setP. Intuitively, this guarantees that $q$ is correctly classified by the nearest-neighbor rule in \setR.

    \item[Case 2 ($\eps$-Approximation guarantee):]
    If $\chromdens{q,\setP} < \eps$.\\
    Let $p \in \setP$ be $q$'s nearest-neighbor in \setP, thus $\dist{q}{p} = \dnn{q,\setP}$. Since \setR is $\alpha$-selective, there exists a point $r \in \setR$ such that $\dist{p}{r} = \dnn{p,\setR} < \dne{p,\setP}/(1+\alpha)$. Additionally, by the triangle inequality and the definition of nearest-enemies, we have 
    \[
        \dne{p,\setP} 
            \leq \dist{p}{\nenemy{q,\setP}} 
            \leq \dist{p}{q} + \dist{q}{\nenemy{q,\setP}} 
            = \dnn{q,\setP}+\dne{q,\setP}.
    \]
    From the definition of chromatic density, $\dne{q,\setP} = (1+\chromdens{q,\setP})\,\dnn{q,\setP}$. Together, these inequalities imply that $(1+\alpha)\,\dist{p}{r} \leq (2+\chromdens{q,\setP})\,\dnn{q,\setP}$. Therefore:
    \begin{equation*}
        \dnn{q,\setR} 
            \leq \dist{q}{r}
            \leq \dist{q}{p} + \dist{p}{r}
            \leq \left( 1 + \frac{2+\chromdens{q,\setP}}{1+\alpha} \right)\dnn{q,\setP}.
    \end{equation*}
    Now, assuming $\chromdens{q,\setP} < \eps$ and setting $\alpha = 2/\eps$, imply that $\dnn{q,\setR} < (1+\eps)\,\dnn{q,\setP}$. Therefore, the nearest-neighbor of $q$ in \setR is an $\eps$-approximate nearest-neighbor of $q$ in \setP.
\end{description}
\vspace*{3pt}
Cases~1 and 2 imply that setting $\alpha = 2/\eps$ is sufficient to ensure that the nearest-neighbor rule classifies any query point $q \in \metricSet$ with the class of one of its $\eps$-approximate nearest-neighbors in \setP. Therefore, \setR is an $\eps$-coreset for the nearest-neighbor rule on \setP.
\end{proof}

So far, we have assumed that nearest-neighbor queries over \setR are computed exactly, as this is the standard notion of coresets. However, it is reasonable to compute nearest-neighbors approximately even for \setR. How should the two approximations be combined to achieve a desired final degree of accuracy? Consider another approximation parameter $\xi$, where $0\leq\xi<\eps$. We say that a set $\setR \subseteq \setP$ is an \emph{$(\xi,\eps)$-coreset} for the approximate nearest-neighbor rule on \setP, if and only if for every query point $q \in \metricSet$, the class of any of $q$'s $\xi$-approximate nearest-neighbor in \setR is the same as the class of any of its $\eps$-approximate nearest-neighbors in \setP. The following result generalizes Theorem~\ref{thm:coreset:nn} to accommodate for $\xi$-\ANN queries after condensation.

\begin{theorem}
\label{thm:coreset:ann}
Any $\alpha$-selective subset is an $(\xi,\eps)$-coreset for the approximate nearest-neighbor rule when $\alpha = \Omega\kern-1pt\left( 1/(\eps - \xi) \right)$.
\end{theorem}

\begin{proof}
This follows from similar arguments to the ones described in the proof of Theorem~\ref{thm:coreset:nn}. Instead, here we set $\alpha = (\eps\kern1pt\xi +3\xi + 2)/(\eps - \xi)$. Let \setR be an $\alpha$-selective subset of \setP, and $q \in \metricSet$ any query point in the metric space, consider the same two cases:
\vspace*{5pt}
\begin{description}
    \setlength{\itemsep}{5pt}
    \item[Case 1 (Correct-Classification guarantee):] If $\chromdens{q,\setP} \geq \eps$.\\Consider the bound derived in Lemma~\ref{lemma:bound-chromdens}. By our assumption that $\chromdens{q,\setP} \geq \eps > 0$, and since $\alpha \geq 0$, the following inequality holds true:
    \[
    \chromdens{q,\setR}
        > \frac{\alpha\,\chromdens{q,\setP} - 2}{\chromdens{q,\setP} + \alpha + 3}
        \geq \frac{\alpha\eps - 2}{\eps + \alpha + 3}
    \]
    Based on this, it is easy to see that the assignment of $\alpha = (\eps\kern1pt\xi +3\xi + 2)/(\eps - \xi)$ implies that $\chromdens{q,\setR} > \xi$, meaning that any of $q$'s $\xi$-approximate nearest-neighbors in \setR belong to the same class as $q$'s nearest-neighbor in \setP. Intuitively, this guarantees that $q$ is correctly classified by the $\xi$-\ANN rule in \setR.

    \item[Case 2 ($\eps$-Approximation guarantee):] If $\chromdens{q,\setP} < \eps$.\\
    The assignment of $\alpha$ implies that $\dnn{q,\setR} < \frac{1+\eps}{1+\xi}\,\dnn{q,\setP}$. This means that an $\xi$-\ANN query for $q$ in \setR, will return one of $q$'s $\eps$-approximate nearest-neighbors in \setP.
\end{description}
\vspace*{5pt}
All together, this implies that \setR is an $(\xi,\eps)$-coreset for the nearest-neighbor rule on \setP.
\end{proof}

In contrast with standard condensation criteria, these new results provide guarantees on either approximation or the correct classification, of any query point in the metric space. This is true even for query points that were ``hard'' to classify with the entire training set, formally defined as query points with low chromatic density. Consequently, Theorems~\ref{thm:coreset:nn}~and~\ref{thm:coreset:ann} show that $\alpha$ must be set to large values if we hope to provide any sort of guarantees for these query points.
However, better results can be achieved by restricting the set of points that are guaranteed to be correctly classified. This relates to the notion of \emph{weak} coresets, which provide approximation guarantees only for a subset of the possible queries. Given $\beta \geq 0$, we define $\mathcal{Q}_\beta$ as the set of query points in $\metricSet$ whose chromatic density with respect to \setP is at least $\beta$ (\ie $\mathcal{Q}_\beta = \{ q \in \metricSet \mid \chromdens{q,\setP} \geq \beta \}$). The following result describes the trade-off between $\alpha$ and $\beta$ to guarantee the correct classification of query points in $\mathcal{Q}_\beta$ after condensation.

\begin{theorem}
\label{thm:coreset:weak}
Any $\alpha$-consistent subset is a weak $\eps$-coreset for the nearest-neighbor rule for queries in $\mathcal{Q}_\beta$, for $\beta = 2/\alpha$. Moreover, all query points in $\mathcal{Q}_\beta$ are correctly classified.
\end{theorem}

The proof of this theorem is rather simple, and follows the same arguments outlined in Case 1 of the proof of Theorem~\ref{thm:coreset:nn}. Basically, we use Lemma~\ref{lemma:bound-chromdens} to show that for any query point $q \in \mathcal{Q}_\beta$, $q$'s chromatic density after condensation is greater than zero if $\alpha \beta \geq 2$. Note that $\eps$ plays no role in this result, as the guarantee on query points of $\mathcal{Q}_\beta$ is of correct classification (\ie the class of its \emph{exact} nearest-neighbor in \setP), rather than an approximation.

The trade-off between $\alpha$ and $\beta$ is illustrated in Figure~\ref{fig:weakcoreset}. From an initial training set $\setP \subset \RE^2$ (Figure~\ref{fig:weakcoreset:dataset}), we show the regions of $\RE^2$ that comprise the sets $\mathcal{Q}_\beta$ for $\beta = 2/\alpha$, using $\alpha = \{ 0.1, 0.2, \sqrt{2} \}$ (Figures~\ref{fig:weakcoreset:20}-\ref{fig:weakcoreset:sqrt2}).
While evidently, increasing $\alpha$ guarantees that more query points will be correctly classified after condensation, this example demonstrates a phenomenon commonly observed experimentally: most query points lie far from enemy points, and thus have high chromatic density with respect to \setP. Therefore, while Theorem~\ref{thm:coreset:nn} states that $\alpha$ must be set to $2/\eps$ to provide approximation guarantees on all query points, Theorem~\ref{thm:coreset:weak} shows that much smaller values of $\alpha$ are sufficient to provide guarantees on some query points, as evidenced in the example in Figure~\ref{fig:weakcoreset}.

\begin{figure*}[h!]
    \vspace*{.2cm}
    \centering
    \begin{subfigure}[b]{.25\linewidth}
        \centering\includegraphics[width=.9\textwidth]{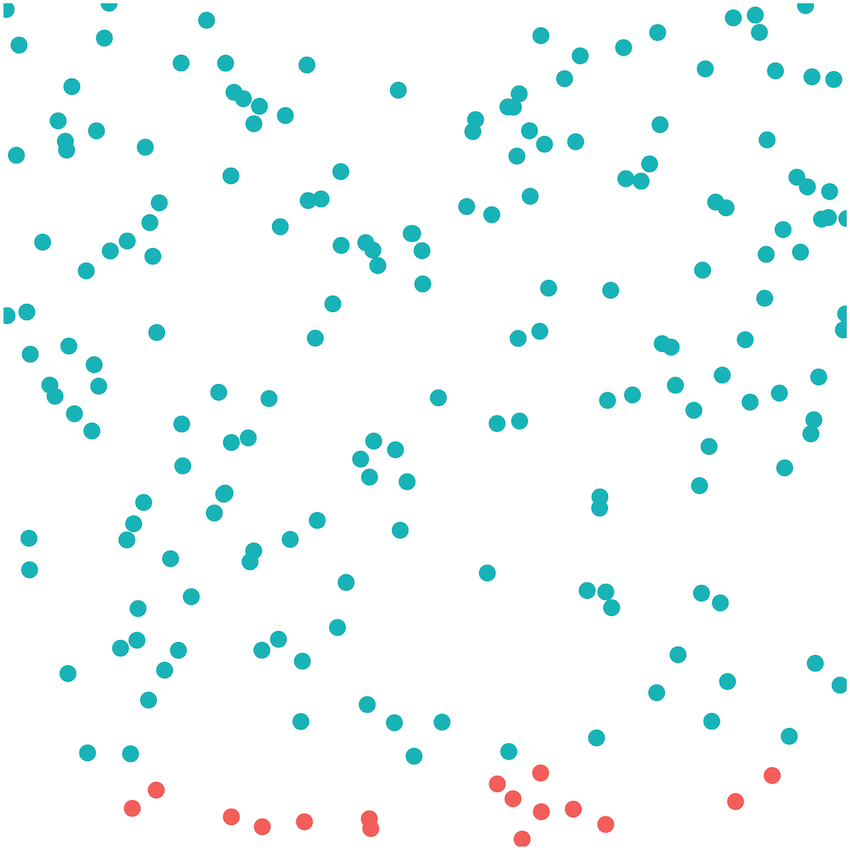}
        \caption{Training set (200\,pts)}\label{fig:weakcoreset:dataset}
    \end{subfigure}%
    \begin{subfigure}[b]{.25\linewidth}
        \centering\includegraphics[width=.9\textwidth]{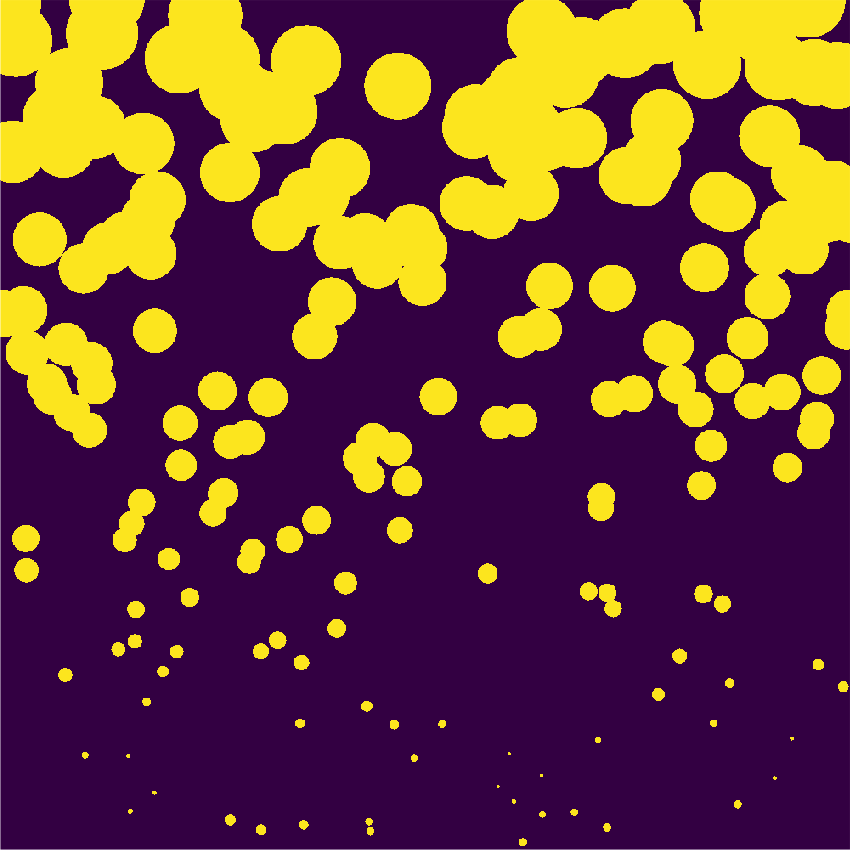}
        \caption{$\mathcal{Q}_{2/\alpha}$ for $\alpha = 0.1$}\label{fig:weakcoreset:20}
    \end{subfigure}%
    \begin{subfigure}[b]{.25\linewidth}
        \centering\includegraphics[width=.9\textwidth]{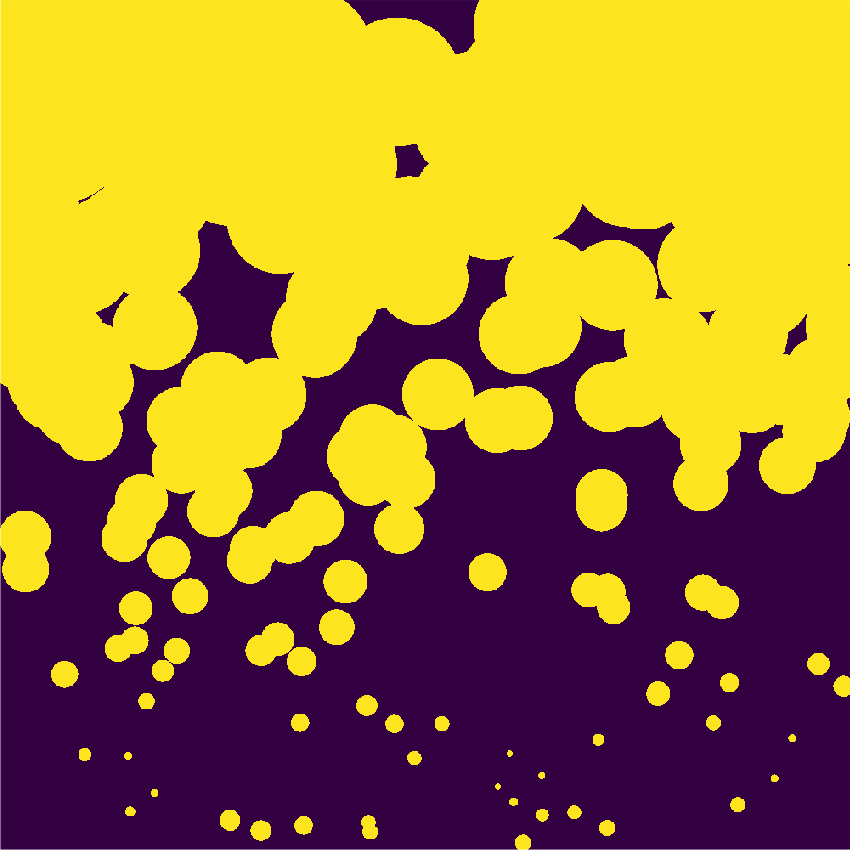}
        \caption{$\mathcal{Q}_{2/\alpha}$ for $\alpha = 0.2$}\label{fig:weakcoreset:10}
    \end{subfigure}%
    \begin{subfigure}[b]{.25\linewidth}
        \centering\includegraphics[width=.9\textwidth]{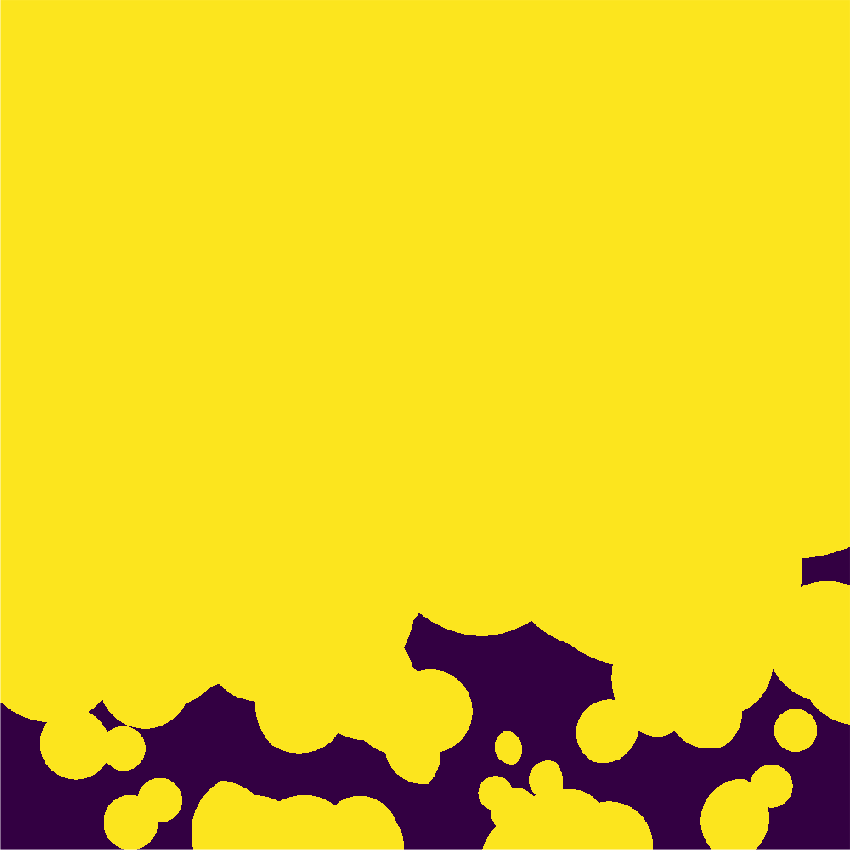}
        \caption{$\mathcal{Q}_{2/\alpha}$ for $\alpha = \sqrt{2}$}\label{fig:weakcoreset:sqrt2}
    \end{subfigure}%
    \caption{Depiction of the $\mathcal{Q}_\beta$ sets for which any $\alpha$-consistent subset is weak coreset ($\beta = 2/\alpha$). Query points in the \emph{yellow}~{\color{yellowcd}$\bullet$} areas are inside $\mathcal{Q}_\beta$, and thus correctly classified after condensation. Query points in the \emph{blue}~{\color{bluecd}$\bullet$} areas are not in $\mathcal{Q}_\beta$, and have no guarantee of correct classification.}\label{fig:weakcoreset}
\end{figure*}

These results establish a clear connection between the problem of condensation and that of finding coresets for the nearest-neighbor rule, and provides a roadmap to prove Theorem~\ref{thm:coreset:main}. This is the first characterization of sufficient conditions to correctly classify any query point in $\metricSet$ after condensation, and not just the points in \setP (as the original consistency criteria implies).
In the following section, these existential results are matched with algorithms to compute $\alpha$-selective subsets of \setP of bounded cardinality.


\section{Coreset Computation}
\label{sec:computation}


\subsection{Hardness Results}
\label{sec:hardness}

Define \textsc{Min-$\alpha$-CS} to be the problem of computing an $\alpha$-consistent subset of minimum cardinality for a given training set \setP. Similarly, let \textsc{Min-$\alpha$-SS} be the corresponding optimization problem for $\alpha$-selective subsets.
Following known results from standard condensation~\cite{Wilfong:1991:NNP:109648.109673,Zukhba:2010:NPP:1921730.1921735,khodamoradi2018consistent}, when $\alpha$ is set to zero, the \textsc{Min-0-CS} and \textsc{Min-0-SS} problems are both known to be NP-hard. Being special cases of the general problems just defined, this implies that both \textsc{Min-$\alpha$-CS} and \textsc{Min-$\alpha$-SS} are NP-hard.

In this section, we present results related to the hardness of approximation of both problems, along with simple algorithmic approaches with tight approximation factors.

\begin{theorem}
\label{thm:hardness:min-alpha-cs}
The \textsc{Min-$\alpha$-CS} problem is \textup{NP}-hard to approximate in polynomial time within a factor of $2^{({\ddim \log{((1+\alpha)/\gamma)})}^{1-o(1)}}$.
\end{theorem}

The full proof is omitted, as it follows from a modification of the hardness bounds proof for the \textsc{Min-0-CS} problem described in~\cite{gottlieb2014near}, which is based on a reduction from the \emph{Label Cover} problem. Proving Theorem~\ref{thm:hardness:min-alpha-cs} involves a careful adjustment of the distances in this reduction, so that all the points in the construction have chromatic density at least $\alpha$. Consequently, this implies that the minimum nearest-enemy distance is reduced by a factor of $1/(1+\alpha)$, explaining the resulting bound for \textsc{Min-$\alpha$-CS}.

The \NET algorithm~\cite{gottlieb2014near} can also be generalized to compute $\alpha$-consistent subsets of \setP as follows. We define \alphaNET as the algorithm that computes a $\gamma/(1+\alpha)$-net of \setP, where $\gamma$ is the smallest nearest-enemy distance in \setP. The covering property of nets~\cite{har2006fast} implies that the resulting subset is $\alpha$-consistent, while the packing property suggests that its cardinality is $\bigOh\left( ((1+\alpha)/\gamma)^{\ddim+1} \right)$, implying a tight approximation to the \textsc{Min-$\alpha$-CS} problem.

\begin{theorem}
\label{thm:hardness:min-alpha-ss}
The \textsc{Min-$\alpha$-SS} problem is \textup{NP}-hard to approximate in polynomial time within a factor of $(1-o(1))\ln{n}$ unless $\textup{NP} \subseteq \textup{DTIME}(n^{\log\log{n}})$.
\end{theorem}


\begin{proof}
The result follows from the hardness of another related covering problem: the mi\-nimum \emph{dominating set}~\cite{feige1998threshold,paz1981non,lund1994hardness}. We describe a simple L-reduction from any instance of this problem to an instance of \textsc{Min-$\alpha$-SS}, which preserves the approximation ratio.

\begin{enumerate}
    \item Consider any instance of minimum dominating set, consisting of the graph $G=(V,E)$.
    \item Generate a new edge-weighted graph $G'$ as follows:\\
        Create two copies of $G$, namely $G_\textsf{r}=(V_\textsf{r},E_\textsf{r})$ and $G_\textsf{b}=(V_\textsf{b},E_\textsf{b})$, of \emph{red} and \emph{blue} nodes respectively. Set all edge-weights of $G_\textsf{r}$ and $G_\textsf{b}$ to be 1. Finally, connect each red node $v_\textsf{r}$ to its corresponding blue node $v_\textsf{b}$ by an edge $\{v_\textsf{r},v_\textsf{b}\}$ of weight $1+\alpha+\xi$ for a sufficienly small constant $\xi>0$. Formally, $G'$ is defined as the edge-weighted graph $G' =(V',E')$ where the set of nodes is $V' = V_\textsf{r} \cup V_\textsf{b}$, the set of edges is $E' = E_\textsf{r} \cup E_\textsf{r} \cup \{ \{v_\textsf{r},v_\textsf{b}\} \mid v \in V \}$, and an edge-weight function $w : E' \rightarrow \RE^+$ where $w(e) = 1$ iff $e \in E_\textsf{r} \cup E_\textsf{b}$, and $w(e) = 1+\alpha+\xi$ otherwise.
    \item A labeling function $l$ 
    where $l(v) = \emph{red}$ iff $v \in V_\textsf{r}$, and $l(v) = \emph{blue}$ iff $v \in V_\textsf{b}$.
    \item Compute the shortest-path metric of $G'$, denoted as $\metricFunc_{G'}$.
    \item Solve the \textsc{Min-$\alpha$-SS} problem for the set $V'$, on metric $\metricFunc_{G'}$, and the labels defined by $l$.
\end{enumerate}

A dominating set of $G$ consists of a subset of nodes $D \subseteq V$, such that every node $v \in V \setminus D$ is adjacent to a node in $D$.
Given any dominating set $D \subseteq V$ of $G$, it is easy to see that the subset $R = \{ v_\textsf{r}, v_\textsf{b} \mid v \in D\}$ is an $\alpha$-selective subset of $V'$, where $|R| = 2|D|$. Similarly, given an $\alpha$-selective subset $R \subseteq V'$, there is a corresponding dominating set $D$ of $G$, where $|D| \leq |R|/2$, as $D$ can be either $R \cap V_\textsf{r}$ or $R \cap V_\textsf{b}$. Therefore, \textsc{Min-$\alpha$-SS} is as hard to approximate as the minimum dominating set problem.
\end{proof}

There is a clear connection between the \textsc{Min-$\alpha$-SS} problem and covering problems, in particular that of finding an optimal hitting set. Given a set of elements $U$ and a family $C$ of subsets~of~$U$, a \emph{hitting set} of $(U,C)$ is a subset $H \subseteq U$ such that every set in $C$ contains at least one element of $H$. Therefore, let $N_{p,\alpha}$ be the set of points of \setP whose distance to $p$ is less~than $\dne{p}/(1+\alpha)$, then any hitting set of $(\setP, \{N_{p,\alpha} \mid p \in \setP\})$ is also an $\alpha$-selective subset of \setP, and vice versa. This simple reduction implies a $\bigOh(n^3)$ worst-case time $\bigOh(\log{n})$-approximation algorithm for \textsc{Min-$\alpha$-SS}, based on the classic greedy algorithm for set cover~\cite{10.2307/3689577,Slavik:1996:TAG:237814.237991}.
Call this approach \alphaHSS or \emph{$\alpha$-Hitting Selective Subset}.
It follows from Theorem~\ref{thm:hardness:min-alpha-ss} that for training sets in general metric spaces, this is the best approximation possible under standard complexity assumptions.

While both \alphaNET and \alphaHSS compute tight approximations of their corresponding problems, their performance in practice does not compare to heuristic approaches for standard condensation (see Section~\ref{sec:experiments} for experimental results). Therefore, in the following sections, we consider two practical algorithms for this problem, namely \FCNN and \RSS, and extend them to compute subsets with the newly defined criteria.


\subsection{An Algorithm for $\alpha$-Selective Subsets}
\label{sec:algorithm:selective}

For standard condensation, the \RSS algorithm was recently proposed~\cite{DBLP:conf/cccg/Flores-VelazcoM19} to compute selective subsets. It runs in quadratic worst-case time and exhibits good performance in practice. The selection process of this algorithm is heuristic in nature and can be described as follows: beginning with an empty set, the points in $p \in \setP$ are examined in increasing order with respect to their nearest-enemy distance $\dne{p}$. The point $p$ is added to the subset \setR if $\dnn{p,\setR} \geq \dne{p}$. It is easy to see that the resulting subset is selective.

We define a generalization, called \alphaRSS, to compute $\alpha$-selective subsets of \setP. The condition to add a point $p \in \setP$ to the selected subset checks if any previously selected point is closer to $p$ than $\dne{p}/(1+\alpha)$, instead of just $\dne{p}$. See Algorithm~\ref{alg:alpha-rss} for a formal description, and Figure~\ref{fig:alpha:rss:process} for an illustration. It is easy to see that this algorithm computes an $\alpha$-selective subset, while keeping the quadratic time complexity of the original \RSS algorithm.

\begin{algorithm}
 \DontPrintSemicolon
 \vspace*{0.1cm}
 \KwIn{Initial training set \setP and parameter $\alpha \geq 0$}
 \KwOut{$\alpha$-selective subset $\setR \subseteq \setP$}
 $\setR \gets \phi$\;
 Let $\left\lbrace p_i \right\rbrace^n_{i=1}$ be the points of \setP sorted increasingly \wrt $\dne{p_i}$\;
 \ForEach{$p_i \in \setP$, where $i = 1\dots n$}{
  \If{$\dnn{p_i, \setR} \geq \dne{p_i}/(1+\alpha)$}{
   $\setR \gets \setR \cup \left\lbrace p_i \right\rbrace$\;
   }
 }
 \KwRet{\setR}
 \vspace*{0.1cm}
 \caption{\alphaRSS}
 \label{alg:alpha-rss}
\end{algorithm}

Naturally, we want to analyze the number of points this algorithm selects. The remainder of this section establishes upper-bounds and approximation guarantees of the \alphaRSS algorithm for any doubling metric space, with improved results in the Euclidean space. This resolves the open problem posed in~\cite{DBLP:conf/cccg/Flores-VelazcoM19} of whether \RSS computes an approximation of the \textsc{Min-0-CS} and \textsc{Min-0-SS}~problems.


\subparagraph*{Size in Doubling spaces.}
First, we consider the case where the underlying metric space \metricSpace of \setP is doubling. The following results depend on the doubling dimension $\ddim$ of the metric space (which is assumed to be constant), the margin $\gamma$ (the smallest nearest-enemy distance of any point in \setP), and $\numNE$ (the number of nearest-enemy points in \setP).

\begin{theorem}
\label{thm:rss-approx-factor-cs}
\alphaRSS computes a tight approximation for the \textsc{Min-$\alpha$-CS} problem.
\end{theorem}

\begin{proof}
This follows from a direct comparison to the resulting subset of the \alphaNET algorithm from the previous section. For any point $p$ selected by \alphaNET, let $B_{p,\alpha}$ be the set of points of \setP ``covered'' by $p$, that is, whose distance to $p$ is at most $\gamma/(1+\alpha)$. By the covering property of $\eps$-nets, this defines a partition on \setP when considering every point $p$ selected by \alphaNET.

Let \setR be the set of points selected by \alphaRSS, we analyze the size of $B_{p,\alpha} \cap \setR$, that is, for any given $B_{p,\alpha}$ how many points could have been selected by the \alphaRSS algorithm. Let $a, b \in B_{p,\alpha} \cap \setR$ be any two such points, where without loss of generality, $\dne{a} \leq \dne{b}$. By the selection process of the algorithm, we know that $\dist{a}{b} \geq \dne{b}/(1+\alpha) \geq \gamma/(1+\alpha)$. A simple packing argument in doubling metrics implies that $|B_{p,\alpha} \cap \setR| \leq 2^{\ddim+1}$. Altogether, we have that the size of the subset selected by \alphaRSS is $\bigOh\left( (2(1+\alpha)/\gamma)^{\ddim+1} \right)$.
\end{proof}

\begin{figure}[t!]
    \centering
    \includegraphics[width=.75\textwidth]{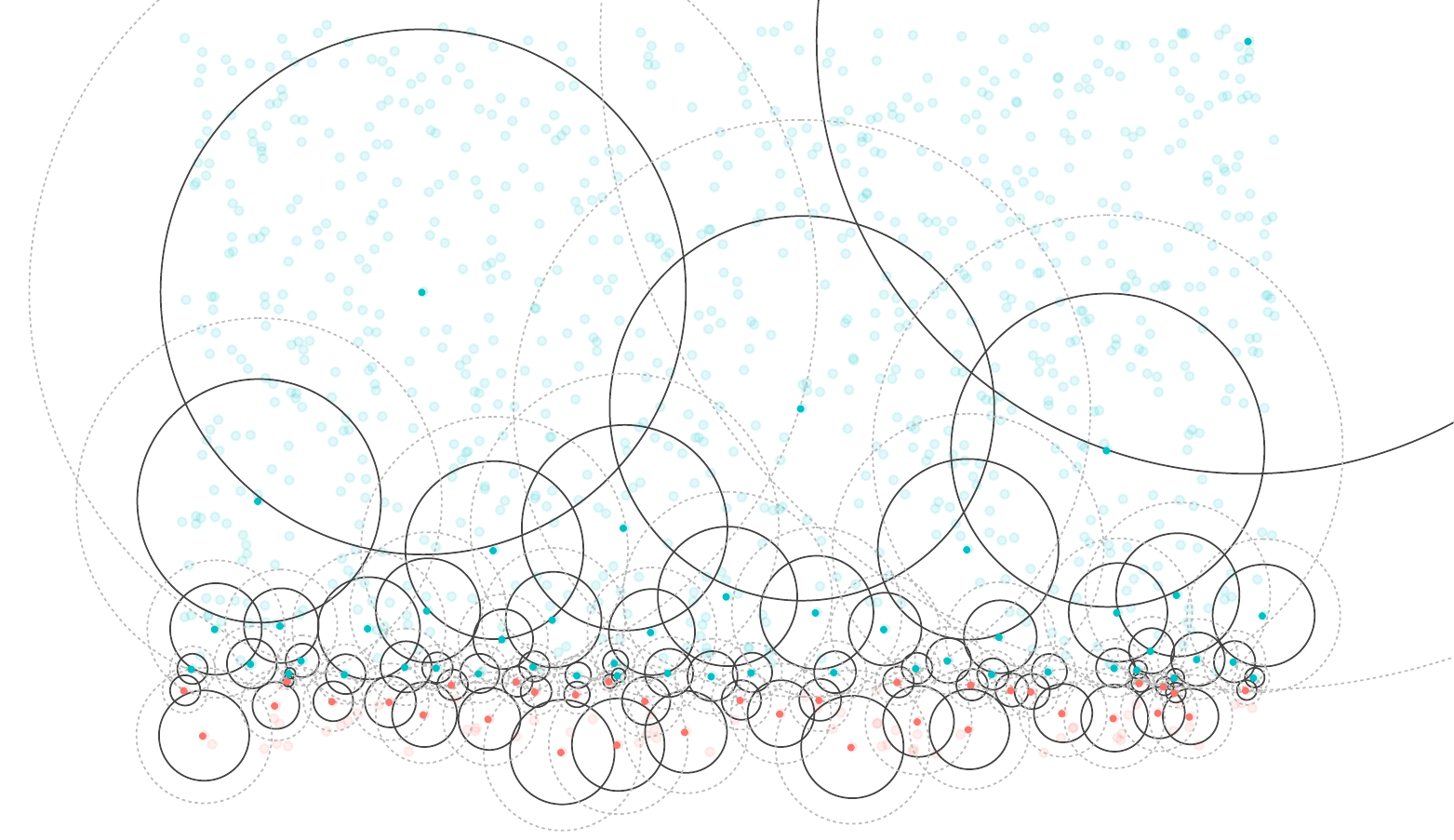}
    \caption{Selection of \alphaRSS for $\alpha\texttt{=}0.5$. Faded points are not selected, while selected points~are drawn along with a ball of radius $\dne{p}$ (dotted outline) and a ball of radius $\dne{p}/(1+\alpha)$ (solid outline). A point $p$ is selected if no previously selected point is closer to $p$ than $\dne{p}/(1+\alpha)$.}
	\label{fig:alpha:rss:process}
\end{figure}

\begin{theorem}
\label{thm:rss-approx-factor}
\alphaRSS computes an $\bigOh\left( \log{(\min{(1+2/\alpha,1/\gamma)})} \right)$-factor approximation for the \textsc{Min-$\alpha$-SS} problem. For $\alpha = \Omega(1)$, this is a constant-factor approximation.
\end{theorem}

\begin{proof}
Let \textup{OPT}$_\alpha$ be the optimum solution to the \textsc{Min-$\alpha$-SS} problem, \ie the minimum cardinality $\alpha$-selective subset of \setP.
For every point $p \in \textup{OPT}_\alpha$ in such solution, define $S_{p,\alpha}$ to be the set of points in \setP ``covered'' by $p$, or simply $S_{p,\alpha} = \{ r \in \setP \mid \dist{r}{p} < \dne{r}/(1+\alpha) \}$.
Additionally, let \setR be the set of points selected by \alphaRSS, define $\setR_{p,\sigma}$ to be the points selected by \alphaRSS which also belong to $S_{p,\alpha}$ and whose nearest-enemy distance is between $\sigma$ and~$2\sigma$, for $\sigma \in [\gamma,1]$. That is, $\setR_{p,\sigma} = \{ r \in \setR \cap S_{p,\alpha} \mid \dne{r} \in [\sigma,2\sigma) \}$. Clearly, these subsets define a partitioning of \setR for all $p \in \textup{OPT}_\alpha$ and values of $\sigma = \gamma\,2^i$ for $i=\{0,1,2,\dots,\lceil\log{\frac{1}{\gamma}}\rceil\}$.

However, depending on $\alpha$, some values of $\sigma$ would yield empty $\setR_{p,\sigma}$ sets. Consider some point $q \in S_{p,\alpha}$, we can bound its nearest-enemy distance with respect to the nearest-enemy distance of point $p$. In particular, by leveraging simple triangle-inequality arguments, it is possible to prove that $\frac{1+\alpha}{2+\alpha} \,\dne{p} \leq \dne{q} \leq \frac{1+\alpha}{\alpha} \,\dne{p}$.
Therefore, the values of $\sigma$ for which $\setR_{p,\sigma}$ sets are not empty, are $\sigma = 2^j \,\frac{1+\alpha}{2+\alpha}\,\dne{p}$ for $j = \{0,\dots,\lceil\log{(1+2/\alpha)}\rceil\}$.

The proof now follows by bounding the size of $\setR_{p,\sigma}$ which can be achieved by bounding its spread. Thus, lets consider the smallest and largest pairwise distances among points in $\setR_{p,\sigma}$.
Take any two points $a,b \in \setR_{p,\sigma}$ where without loss of generality, $\dne{a} \leq \dne{b}$. Note that points selected by \alphaRSS cannot be ``too close'' to each other; that is, as $a$ and $b$ were selected by the algorithm, we know that $(1+\alpha)\,\dist{a}{b} \geq \dne{b} \geq \sigma$. Therefore, the smallest pairwise distance in $\setR_{p,\sigma}$ is at least $\sigma/(1+\alpha)$.
Additionally, by the triangle inequality, we can bound the maximum pairwise distance using their distance to $p$ as $\dist{a}{b} \leq \dist{a}{p} + \dist{p}{b} \leq 4\sigma / (1+\alpha)$.
Then, by the packing properties of doubling spaces, the size of $\setR_{p,\sigma}$ is at most $4^{\ddim+1}$.

Altogether, for every $p \in \textup{OPT}_\alpha$ there are up to $\lceil\log{(\min{(1+2/\alpha,1/\gamma)})}\rceil$ non-empty $\setR_{p,\sigma}$ subsets, each containing at most $4^{\ddim+1}$ points.
In doubling spaces with constant doubling dimension, the size of these subsets is also constant.
\end{proof}



While these results are meaningful from a theoretical perspective, it is also useful to establishing bounds in terms of the geometry of the learning space, which is characterized by the boundaries between points of different classes. Thus, using similar packing arguments as above, we bound the selection size of the algorithm with respect to $\numNE$.

\begin{theorem}
\label{thm:rss-size}
\alphaRSS selects $\bigOh\left( \numNE \log{\frac{1}{\gamma}}\ (1+\alpha)^{\ddim+1}\right)$ points.
\end{theorem}

\begin{proof}
This follows from similar arguments to the ones used to prove Theorem~\ref{thm:rss-approx-factor}, using an alternative charging scheme for each nearest-enemy point in the training set.
Consider one such point $p \in \{ \nenemy{r} \mid r \in \setP\}$ and a value $\sigma \in [\gamma,1]$, we define $\setR'_{p,\sigma}$ to be the subset of points from \alphaRSS whose nearest-enemy is $p$, and their nearest-enemy distance is between $\sigma$ and $2\sigma$. That is, $\setR'_{p,\sigma} = \{ r \in \setR \mid \nenemy{r}=p \wedge \dne{r} \in [\sigma,2\sigma) \}$. These subsets partition \setR for all nearest-enemy points of \setP, and values of $\sigma = \gamma\, 2^i$ for $i=\{0,1,2,\dots,\lceil\log{\frac{1}{\gamma}}\rceil\}$.

For any two points $a, b \in \setR'_{p,\sigma}$, the selection criteria of \alphaRSS implies some separation between selected points, which can be used to prove that $\dist{a}{b} \geq \sigma/(1+\alpha)$. Additionally, we know that $\dist{a}{b} \leq \dist{a}{p} + \dist{p}{b} = \dne{a} + \dne{b} \leq 4\sigma$. Using a simple packing argument, we have that $|\setR'_{p,\sigma}| \leq \lceil 4(1+\alpha) \rceil^{\ddim+1}$.

Altogether, by counting all sets $\setR'_{p,\sigma}$ for each nearest-enemy in the training set and values of $\sigma$, the size of \setR is upper-bounded by $|\setR| \leq \numNE \left\lceil\log{1/\gamma} \right\rceil \left\lceil 4(1+\alpha) \right\rceil^{\ddim+1}$. Based on the assumption that $\ddim$ is constant, this completes the proof.
\end{proof}

As a corollary, this result implies that when $\alpha = 2/\eps$, the $\alpha$-selective subset computed by \alphaRSS contains $\bigOh\left( \numNE \log{1/\gamma}\ (1/\eps)^{\ddim+1} \right)$ points. This establishes the size bound on the $\eps$-coreset given in Theorem~\ref{thm:coreset:main}, which can be computed using the \alphaRSS algorithm.


\subparagraph*{Size in Euclidean space.}
In the case where $\setP \subset \RE^d$ lies in $d$-dimensional Euclidean~space, the analysis of \alphaRSS can be further improved, leading to a constant-factor approximation of \textsc{Min-$\alpha$-SS} for any value of $\alpha \geq 0$, and reduced dependency on the dimensionality of \setP.

\begin{theorem}
\label{thm:rss-approx-factor-euclidean}
\alphaRSS computes an $\bigOh(1)$-approximation for the \textsc{Min-$\alpha$-SS} problem in $\RE^d$.
\end{theorem}

\begin{proof}
Similar to the proof of Theorem~\ref{thm:rss-approx-factor}, define $\setR_p = S_{p,\alpha} \cap \setR$ as the points selected by \alphaRSS that are ``covered'' by $p$ in the optimum solution \textup{OPT}$_\alpha$. Consider two such points $a, b \in \setR_p$ where without loss of generality, $\dne{a} \leq \dne{b}$.
By the definition~of $S_{p,\alpha}$ we know that $\dist{a}{p} < \dne{a}/(1+\alpha)$, and similarly with $b$.
Additionally, from the selection of the algorithm we know that $\dist{a}{b} \geq \dne{b}/(1+\alpha)$.
Overall, these inequalities imply that the angle $\angle apb \geq \pi/3$.
By a simple packing argument, the size of $\setR_p$ is bounded by the \emph{kissing number} in $d$-dimensional Euclidean space, or simply $\bigOh((3/\pi)^{d-1})$. \mbox{Therefore, we have} that $|\setR| \leq \sum_p |\setR_p| = |\textup{OPT}_\alpha|\ \bigOh((3/\pi)^{d-1})$. Assuming $d$ is constant, this completes the proof.
\end{proof} 

Furthermore, a similar constant-factor approximation can be achieved for any training set \setP in $\ell_p$ space for $p\geq 3$. This follows analogously to the proof of Theorem~\ref{thm:rss-approx-factor-euclidean}, exploiting the bounds between $\ell_p$ and $\ell_2$ metrics, where $1/\sqrt{d}\ \lVert v \rVert_p \leq \lVert v \rVert_2 \leq \sqrt{d}\ \lVert v \rVert_p$. This would imply that the angle between any two points in $\alphaRSS_{p}$ is $\Omega(1/d)$. Therefore, it shows that \alphaRSS achieves an approximation factor of $\bigOh(d^{d-1})$, or simply $\bigOh(1)$ for constant dimension.

Similarly to the case of doubling spaces, we also establish upper-bounds in terms of $\numNE$ for the selection size of the algorithm in Euclidean space. The following result improves the exponential dependence on the dimensionality of \setP (from $\textsf{ddim}(\RE^d) = \Theta(d)$ to $d-1$), while keeping the dependency on the margin $\gamma$, which contrast with the approximation~factor~results. 

\begin{theorem}
\label{thm:rss-size-euclidean}
In Euclidean space $\RE^d$, \alphaRSS selects $\bigOh\left( \numNE \log{\frac{1}{\gamma}}\ (1+\alpha)^{d-1} \right)$ points.
\end{theorem}

\begin{proof}
Let $p$ be any nearest-enemy point of \setP and $\sigma \in [\gamma,1]$, similarly define $\setR'_{p,\sigma}$ to be the set of points selected by \alphaRSS whose nearest-enemy is $p$ and their nearest-enemy distance is between $\sigma$ and $b\sigma$, for $b = \frac{(1+\alpha)^2}{\alpha(2+\alpha)}$. Equivalently, these subsets define a partitioning of \setR for all nearest-enemy points $p$ and values of $\sigma = \gamma\,b^k$ for $k=\{0,1,2,\dots,\lceil\log_b{\frac{1}{\gamma}}\rceil\}$.~Thus, the proof follows from bounding the minimum angle between points in these subsets. For any two such points $p_i, p_j \in \setR'_{p,\sigma}$, we lower bound the angle $\angle p_i p p_j$. Assume without loss of generality that $\dne{p_i} \leq \dne{p_j}$. By definition of the partitioning, we also know that $\dne{p_j} \leq b\sigma \leq b \, \dne{p_i}$. Therefore, altogether we have that $\dne{p_i} \leq \dne{p_j} \leq b \, \dne{p_i}$.

First, consider the set of points whose distance to $p_i$ is $(1+\alpha)$ times their distance~to~$p$, which defines a multiplicative weighted bisector~\cite{AURENHAMMER1984251} between points $p$ and $p_i$, with weights equal to $1$ and $1/(1+\alpha)$ respectively. This is characterized as a $d$-dimensional ball (see Figure~\ref{fig:proof:euclidean:bisector}) with center $c_i = (p_i-p)\, b + p$ and radius $\dne{p_i} \, b/(1+\alpha)$. Thus $p$, $p_i$ and $c_i$ are collinear, and the distance between $p$ and $c_i$ is $\dist{p}{c_i} = b \, \dne{p_i}$.
In particular, let's consider the relation between $p_j$ and such bisector. As $p_j$ was selected by the algorithm after $p_i$, we know that $(1+\alpha)\,\dist{p_j}{p_i} \geq \dne{p_j}$ where $\dne{p_j} = \dist{p_j}{p}$. Therefore, clearly $p_j$ lies either outside or in the surface of the weighted bisector between $p$ and $p_i$ (see Figure~\ref{fig:proof:euclidean:general}).

\begin{figure*}[h!]
    \centering
    \begin{subfigure}[b]{.237\linewidth}
        \centering\includegraphics[width=\textwidth]{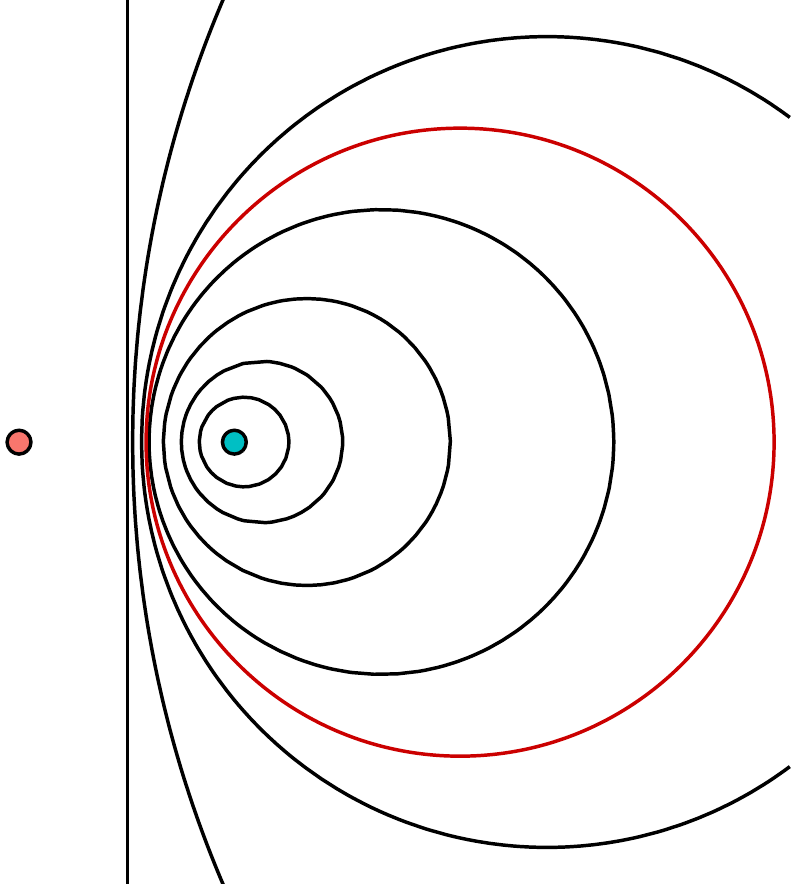}
        \caption{Multiplicatively weighted bisectors for different weights.}
		\label{fig:proof:euclidean:bisector}
    \end{subfigure}\hfill%
    \begin{subfigure}[b]{.356\linewidth}
        \centering\includegraphics[width=\textwidth]{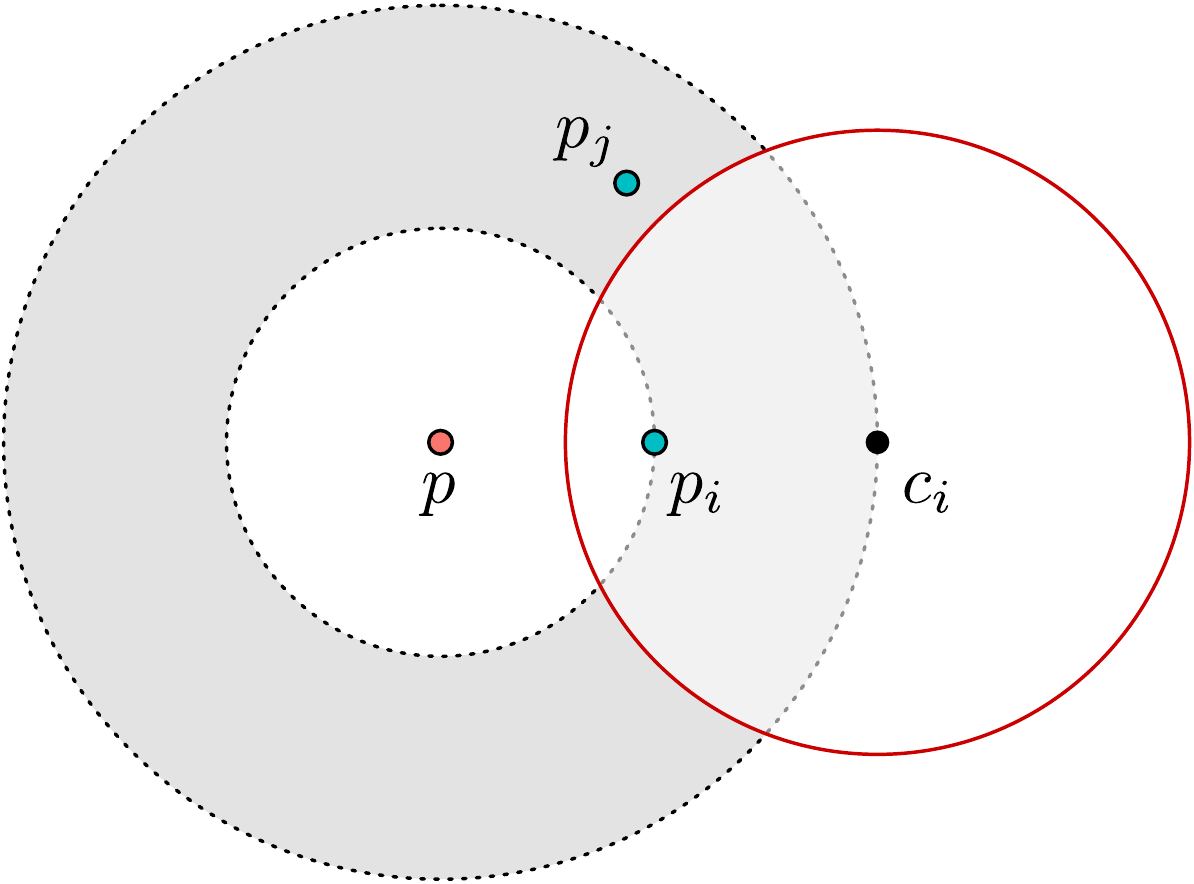}
        \caption{Position of point $p_j$\\ \wrt the weighted bisector\\ between points $p$ and $p_i$.}
		\label{fig:proof:euclidean:general}
    \end{subfigure}\hfill%
    \begin{subfigure}[b]{.356\linewidth}
        \centering\includegraphics[width=\textwidth]{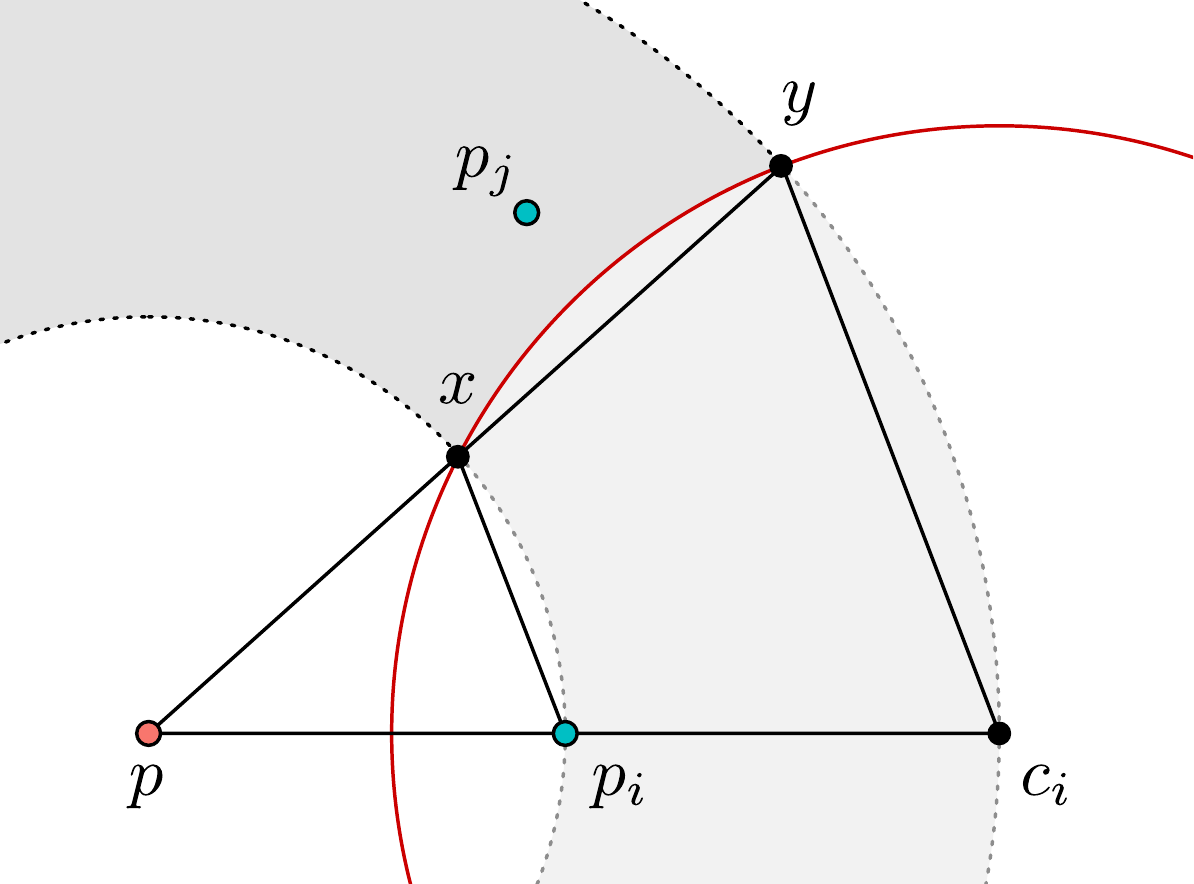}
        \caption{The intersection points $x$ and $y$ between the weighted bisector and the limit balls of $\setR_{p,\sigma}$.}
		\label{fig:proof:euclidean:detail}
    \end{subfigure}%
    \caption{Construction for the analysis of the minimum angle between two points in $\setR'_{p,\sigma}$ \wrt some nearest-enemy point $p \in \setP$. Let points $p_i, p_j \in \setR'_{p,\sigma}$, we analyze the angle $\angle p_i p p_j$.}\label{fig:proof:euclidean}
\end{figure*}

For angle $\angle p_i p p_j$, we can frame the analysis to the plane defined by $p$, $p_i$ and $p_j$. Let $x$ and $y$ be two points in this plane, such that they are the intersection points between the weighted bisector and the balls centered at $p$ of radii $\dne{p_i}$ and $b\, \dne{p_i}$ respectively (see Figure~\ref{fig:proof:euclidean:detail}). By the convexity of the weighted bisector between $p$ and $p_i$, we can say that $\angle p_i p p_j \geq \min(\angle x p p_i , \angle y p c_j)$. Now, consider the triangles $\triangle p x p_i$ and $\triangle p y c_i$. By the careful selection of $b$, these triangles are both isosceles and similar. In particular, for $\triangle p x p_i$ the two sides incident to $p$ have length equal to $\dne{p_i}$, and the side opposite to $p$ has length equal to $\dne{p_i}/(1+\alpha)$. For $\triangle p y c_i$, the side lengths are $b\,\dne{p_i}$ and $\dne{p_i}\, b/(1+\alpha)$. Therefore, the angle $\angle p_i p p_j \geq \angle x p p_i \geq 1/(1+\alpha)$.

By a simple packing argument based on this minimum angle, we have that the size of $\setR'_{p,\sigma}$ is $\bigOh((1+\alpha)^{d-1})$. All together, following the defined partitioning, we have that:
\[
|\setR|
    = \sum_{p} \sum_{k=0}^{\lceil\log_b\frac{1}{\gamma}\rceil} |\setR'_{p,b^k}|
    \leq \numNE \left\lceil\log_b \frac{1}{\gamma} \right\rceil \bigOh\left((1+\alpha)^{d-1} \right)
\]

For constant $\alpha$ and $d$, the size of \alphaRSS is $\bigOh(\numNE \log{\frac{1}{\gamma}})$. Moreover, when $\alpha$ is zero \alphaRSS selects $\bigOh(\numNE\ c^{d-1})$, matching the previously known bound for \RSS in Euclidean space.
\end{proof}


\subsection{Subquadratic Algorithm}
\label{sec:subquadratic}

In this section we present a subquadratic implementation for the \alphaRSS algorithm, which completes the proof of our main result, Theorem~\ref{thm:coreset:main}. Prior to this result    , among algorithms for nearest-neighbor condensation, \FCNN achieves the best worst-case time complexity, running in $\bigOh(nm)$ time, where $m = |\setR|$ is the size of the selected subset. 


The \alphaRSS algorithm consists of two main stages: computing the nearest-enemy distances of all points in \setP (and sorting the points based on these), and the selection process itself.
The first stage requires a total of $n$ nearest-enemy queries, plus additional $\bigOh(n\log{n})$ time for sorting. The second stage performs $n$ nearest-neighbor queries on the current selected subset \setR, which needs to be updated $m$ times. In both cases, using exact nearest-neighbor search would degenerate into linear search due to the \emph{curse of dimensionality}. Thus, the first and second stage of the algorithm would need $\bigOh(n^2)$ and $\bigOh(nm)$ worst-case time respectively.

These bottlenecks can be overcome by leveraging approximate nearest-neighbor techniques. Clearly, the first stage of the algorithm can be improved by computing nearest-enemy distances approximately, using as many \ANN structures as classes there are in \setP, which is considered to be a small constant.
Therefore, by also applying a simple brute-force search for nearest-neighbors in the second stage, result (i) of the next theorem follows immediately. Moreover, by combining this with standard techniques for static-to-dynamic conversions~\cite{bentley1980decomposable}, we have result (ii) below. Denote this variant of \alphaRSS as \paramRSS{(\alpha,\xi)}, for a parameter $\xi \geq 0$.

\begin{theorem}
\label{thm:subquadratic:general}
Given a data structure for $\xi$-\ANN searching with construction time $t_c$ and query time $t_q$ (which potentially depend on $n$ and $\xi$), the \paramRSS{(\alpha,\xi)} variant can be implemented with the following worst-case time complexities, where $m$ is the size of the~selected~subset.
\begin{romanenumerate}
    \item $\bigOh \left( t_c + n\,(t_q + m + \log{n}) \right)$
    \item $\bigOh \left( (t_c + n\,t_q) \log{n} \right)$
\end{romanenumerate}
\end{theorem}

More generally, if we are given an additional data structure for dynamic $\xi$-\ANN searching with construction time $t'_c$, query time $t'_q$, and insertion time $t'_i$, the overall running time will be $\bigOh \left( t_c + t'_c + n\,(t_q + t'_q + \log{n}) + m\,t'_i \right)$. Indeed, this can be used to obtain (ii) from the static-to-dynamic conversions~\cite{bentley1980decomposable}, which propose an approach to convert static search structures into dynamic ones. These results directly imply implementations of \paramRSS{(\alpha,\xi)} with subquadratic worst-case time complexities, based on \ANN techniques~\cite{arya2009space,arya2018approximate} for low-dimensional Euclidean space, and using techniques like LSH~\cite{andoni2018approximate} that are suitable for \ANN in high-dimensional Hamming and Euclidean spaces. More generally, subquadratic runtimes can be achieved by leveraging techniques~\cite{cole2006searching} for dynamic \ANN search in doubling spaces.

\subparagraph*{Dealing with uncertainty.}
Such implementation schemes for \alphaRSS would incur an approximation error (of up to $1+\xi$) on the computed distances: either only during the first stage if (i) is implemented, or during both stages if (ii) or the dynamic-structure scheme are implemented.
The uncertainty introduced by these approximate queries, imply that in order to guarantee finding $\alpha$-selective subsets, we must modify the condition for adding point during the second stage of the algorithm. Let $\dne{p,\xi}$ denote the $\xi$-approximate nearest-enemy distance of $p$ computed in the first stage, and let $\dnn{p,\setR,\xi}$ denote the $\xi$-approximate nearest-neighbor distance of $p$ over points of the current subset (computed in the second stage). Then, \paramRSS{(\alpha,\xi)} adds a point $p$ into the subset if $(1+\xi)(1+\alpha)\,\dnn{p,\setR,\xi} \geq \dne{p,\xi}$.

By similar arguments to the ones described in Section~\ref{sec:algorithm:selective}, size guarantees can be extended to \paramRSS{(\alpha,\xi)}. First, the size of the subset selected by \paramRSS{(\alpha,\xi)}, in terms of the number of nearest-enemy points in the set, would be bounded by the size of the subset selected by \paramRSS{\hat{\alpha}} with $\hat{\alpha} = (1+\alpha)(1+\xi)^2-1$. Additionally, the approximation factor of \paramRSS{(\alpha,\xi)} in both doubling and Euclidean metric  spaces would increase by a factor of $\bigOh((1+\xi)^{2(\ddim+1)})$.

\vspace*{5pt}\noindent This completes the proof of Theorem~\ref{thm:coreset:main}.


\subsection{An Algorithm for $\alpha$-Consistent Subsets}
\label{sec:algorithm:consistent}

Even thought the main result of this paper relies on the computation of $\alpha$-selective subsets, Theorem~\ref{thm:coreset:weak} shows that even $\alpha$-consistency is enough to guarantee the correct classification of certain query points. In practice, \FCNN~\cite{angiulli2007fast} is the most efficient algorithm for computing consistent subsets. Therefore, in this section, we discuss a simple extension of this algorithm in order to compute $\alpha$-consistent subsets.

Recent efforts~\cite{cccg20afloresv} show the first theoretical analysis on the selection size of \FCNN. The results are two fold: while the size of the subset selected by \FCNN cannot be upper-bounded, a simple modification of the algorithm is sufficient to obtain provable upper-bounds. This modified algorithm is called \SFCNN.

Both algorithms, \FCNN and \SFCNN, select points iteratively as follows (see Algorithm~\ref{alg:alpha-sfcnn}). First, the subset \setR is initialized with one point per class (\eg the centroids of each class). During every iteration, the algorithm identifies all the points in \setP that are incorrectly classified with the current \setR, or simply, those whose nearest-neighbor in \setR is of different class. This is formalized as the \emph{voren} function, defined for every point $p \in \setR$ as follows:
\[
    \textup{voren}(p,\setR,\setP)
        = \{ q \in \setP \mid \nn{q,\setR} = p \wedge l(q) \neq l(p) \}
\]

This function identifies all the enemies of $p$ whose nearest-neighbor in \setR is $p$ itself. The only difference between the original \FCNN algorithm and the modified \SFCNN appears next. While \FCNN adds one point per each $p \in \setR$ in a batch%
\footnote{For \FCNN, line 4 of Algorithm~\ref{alg:alpha-sfcnn} updates \setR by adding all the points in $S$, instead of only one point of $S$.}%
, potentially doubling the size of \setR, \SFCNN adds only \emph{one} point per iteration. Then, both algorithms terminate when no other points can be added (\ie all $\textup{voren}(p,\setR,\setP)$ are empty), implying that \setR is consistent.

We can now extend both algorithms to compute $\alpha$-consistent subsets, namely \alphaFCNN and \alphaSFCNN, by redefining the \emph{voren} function. The idea is to identify those points whose nearest-neighbor in \setR is $p$, such that are either enemies of $p$, or whose chromatic density with respect to \setR is less than $\alpha$. This is formally defined as follows:
\[
    \textup{voren}_\alpha(p,\setR,\setP)
        = \{ q \in \setP \mid \nn{q,\setR} = p \wedge (l(q) = l(p) \Rightarrow \chromdens{q,\setR} < \alpha ) \}
\]

By plugging this function into the algorithms (see Algorithm~\ref{alg:alpha-sfcnn}), it is easy to show that the resulting subsets are $\alpha$-consistent. Moreover, this can be easily implemented to run in $\bigOh(nm)$ worst-case time, where $m$ is the final size of \setR, extending the implementation scheme described in the paper where \FCNN was initially proposed~\cite{angiulli2007fast}.

\begin{algorithm}
 \DontPrintSemicolon
 \vspace*{0.1cm}
 \KwIn{Initial training set \setP and parameter $\alpha \geq 0$}
 \KwOut{$\alpha$-consistent subset $\setR \subseteq \setP$}
 $\setR \gets \phi$\;
 $S \gets \textup{centroids}(\setP)$\;
 \While{$S \neq \phi$}{
  $\setR \gets \setR \cup \{ \text{Choose one point from } S\}$\; 
  $S \gets \phi$\;
  \ForEach{$p \in \setR$}{
   $S \gets S \cup \{ \text{Choose one point from } \textup{voren}_\alpha(p,\setR,\setP) \}$\;
  }
 }
 \KwRet{\setR}
 \vspace*{0.1cm}
 \caption{\alphaSFCNN}
 \label{alg:alpha-sfcnn}
\end{algorithm}

Finally, leveraging the analysis described in~\cite{cccg20afloresv}, together with the proofs of Theorems~\ref{thm:rss-approx-factor-cs} and~\ref{thm:rss-size}, we upper-bound the selection size of the \alphaSFCNN algorithm.
The proofs of the next results depend on the following observation. Let $a,b \in \setR$ be two points selected by \alphaSFCNN, where $\dne{a}, \dne{b} \geq \beta$ for some $\beta \geq 0$, it is easy to show that $\dist{a}{b} \geq \beta/(1+\alpha)$. This follows from a fairly simple argument: to the contrary, suppose that $\dist{a}{b} < \beta/(1+\alpha)$, which would imply that $a$ and $b$ belong to the same class. Without loss of generality, point $a$ was added to \setR before point $b$. Note that after adding point $a$ to \setR, the chromatic density of $b$ \wrt \setR is $\chromdens{b,\setR} > \alpha$, which contradicts the statement that $b$ could be added to \setR.

\begin{theorem}
\label{thm:sfcnn-approx-factor}
\alphaSFCNN computes a tight approximation for the \textsc{Min-$\alpha$-CS} problem.
\end{theorem}

This result follows by similar arguments as the proof of Theorem~\ref{thm:rss-approx-factor-cs}. By considering any two points $a,b \in B_{p,\alpha} \cap \setR$, we know that $\dist{a}{b} \geq \gamma/(1+\alpha)$ as $\gamma$ is the smallest nearest-enemy distance in \setP. This implies \alphaSFCNN can select up to $2^{\ddim+1}$ times more points as the \alphaNET algorithm, which yields the proof.

\begin{theorem}
\label{thm:sfcnn-size}
\alphaSFCNN selects $\bigOh\left( \numNE \log{\frac{1}{\gamma}}\ (1+\alpha)^{\ddim+1}\right)$ points.
\end{theorem}

Similarly, this result can be proven using the same arguments outlined to prove Theorem~\ref{thm:rss-size}. After partitioning the selection of \alphaSFCNN into $\bigOh(\numNE \log{1/\gamma})$ subsets, consider any two points $a,b$ in one of these subsets, where $\dne{a}, \dne{b} \in [\sigma, 2\sigma)$, for some $\sigma \in [\gamma,1]$. Therefore, we can show that $\dist{a}{b} \geq \sigma/(1+\alpha)$, which implies that each subset in the partitioning contains at most $\lceil 4(1+\alpha) \rceil^{\ddim+1}$ points. This yields the proof.

\section{Experimental Evaluation}
\label{sec:experiments}

In order to get a clearer impression of the relevance of these results in practice, we performed experimental trials on several training sets, both synthetically generated and widely used benchmarks. First, we consider 21 training sets from the UCI \emph{Machine Learning Repository}\footnote{\url{https://archive.ics.uci.edu/ml/index.php}} which are commonly used in the literature to evaluate condensation algorithms~\cite{Garcia:2012:PSN:2122272.2122582}. These consist of a number of points ranging from 150 to $58000$, in $d$-dimensional Euclidean space with $d$ between 2 and 64, and 2 to 26 classes.
We also generated some synthetic training sets, containing $10^5$ uniformly distributed points, in 2 to 3 dimensions, and 3 classes. 
All training sets used in these experimental trials are summarized in Table~\ref{table:data}. The implementation of the algorithms, training sets used, and raw results, are publicly available\footnote{\url{https://github.com/afloresv/nnc/}}.

These experimental trials compare the performance of different condensation algorithms when applied to vastly different training sets. We use two measures of comparison on these algorithms: their runtime in the different training sets, and the size of the subset selected. Clearly, these values might differ greatly on training sets whose size are too distinct. Therefore, before comparing the raw results, these are normalized. The runtime of an algorithm for a given training set is normalized by dividing it by $n$, the size of the training set. The size of the selected subset is normalized by dividing it by $\numNE$, the number of nearest-enemy points in the training set, which characterizes the complexity of the boundaries between classes.

\subparagraph*{Algorithm Comparison.}
The first experiment evaluates the performance of the five algorithms discussed in this paper: \alphaRSS, \alphaFCNN, \alphaSFCNN, \alphaHSS, and \alphaNET. The evaluation is carried out by varying the value of the $\alpha$ parameter from 0 to 1, to understand the impact of increasing this parameter. The implementation of \alphaHSS uses the well-known greedy algorithm for set cover~\cite{10.2307/3689577}, and solves the problem using the reduction described in Section~\ref{sec:hardness}. In the other hand, recall that the original \NET algorithm (for $\alpha=0$) implements an extra pruning technique to further reduce the training set after computing the $\gamma$-net~\cite{gottlieb2014near}. For a fair comparison, we implemented the \alphaNET algorithm with a modified version of this pruning technique that guarantees that the selected subset is still $\alpha$-selective.

The results show that \alphaRSS outperforms the other algorithms in terms of running time by a big margin, and irrespective of the value of $\alpha$ (see Figure~\ref{fig:exp:alpha:time}). Additionally, the number of points selected by \alphaRSS, \alphaFCNN, and \alphaSFCNN is comparable to \alphaHSS, which guarantees the best possible approximation factor in general metrics, while \alphaNET is significantly outperformed.

\begin{figure*}[h!]
    \centering
    \begin{subfigure}[b]{.48\linewidth}
        \centering\includegraphics[width=\textwidth]{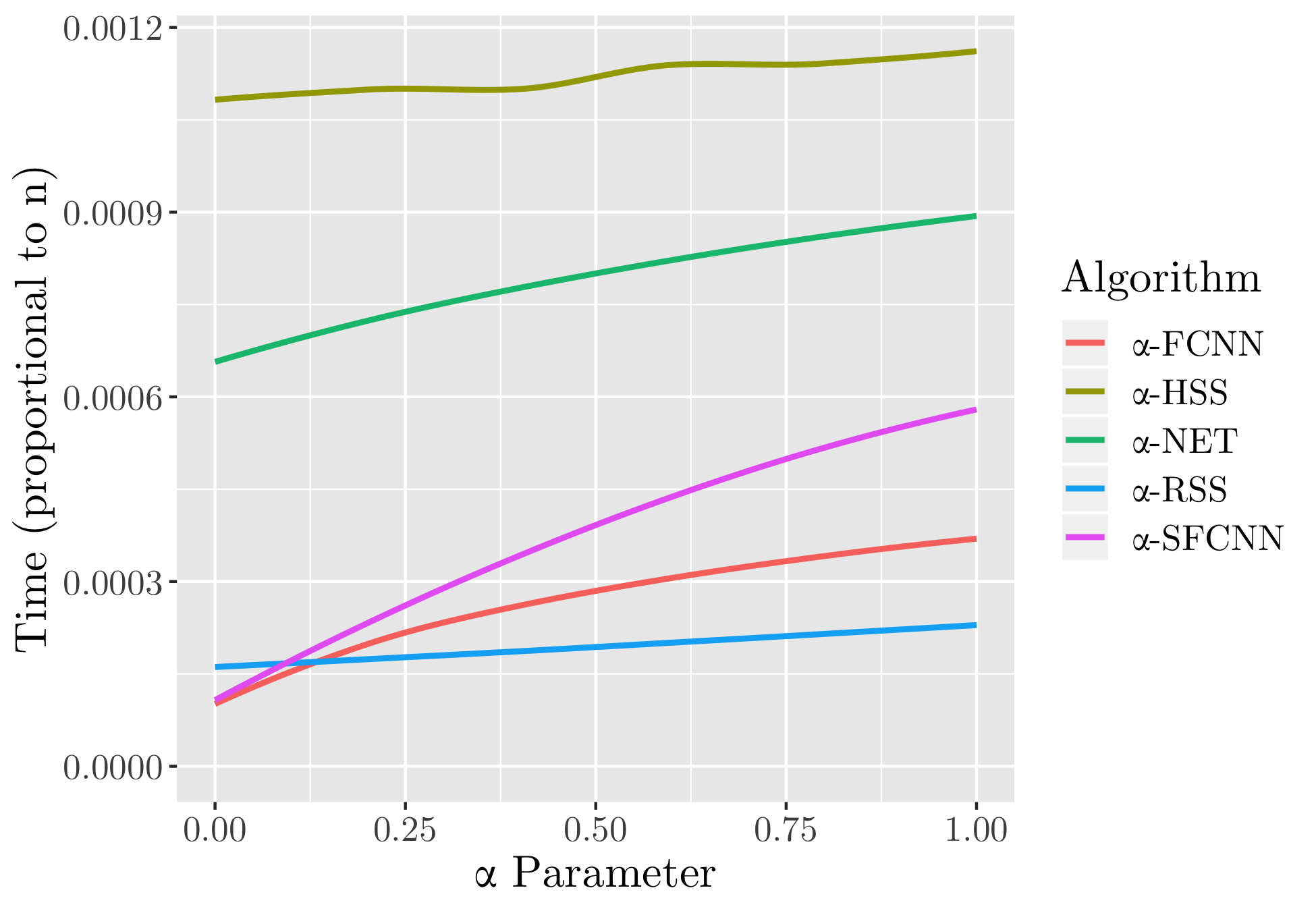}
        \vspace*{-10pt}
        \caption{Running time.}\label{fig:exp:alpha:time}
    \end{subfigure}%
    \hfill
    \begin{subfigure}[b]{.48\linewidth}
        \centering\includegraphics[width=\textwidth]{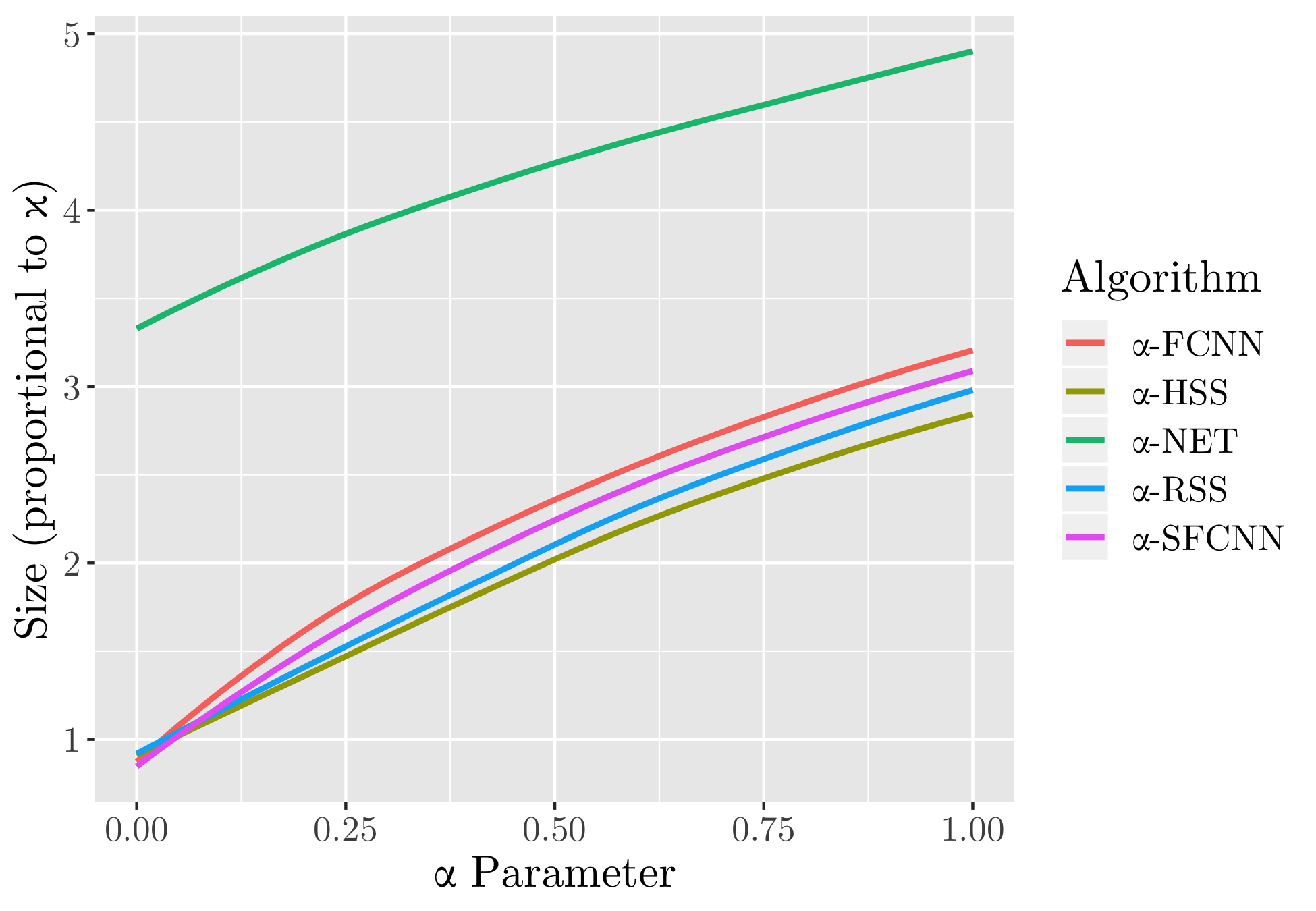}
        \vspace*{-10pt}
        \caption{Size of the selected subsets.}\label{fig:exp:alpha:size}
    \end{subfigure}%
    \vspace*{-5pt}
    \caption{Comparison \alphaRSS, \alphaFCNN, \alphaSFCNN, \alphaNET, and \alphaHSS, for different values of $\alpha$.}\label{fig:exp:alpha}
\end{figure*}

\subparagraph*{Subquadratic Approach.}
Using the same experimental framework, we evaluate performance of the subquadratic implementation \paramRSS{(\alpha,\xi)} described in Section~\ref{sec:subquadratic}. In this case, we change the value of parameter $\xi$ to assess its effect on the running time and selection size over the algorithm, for two different values of $\alpha$ (see Figure~\ref{fig:exp:eps}). The results show an expected increase of the number of selected points, while significantly improving its running time.

\begin{figure*}[h!]
    \centering
    \begin{subfigure}[b]{.48\linewidth}
        \centering\includegraphics[width=\textwidth]{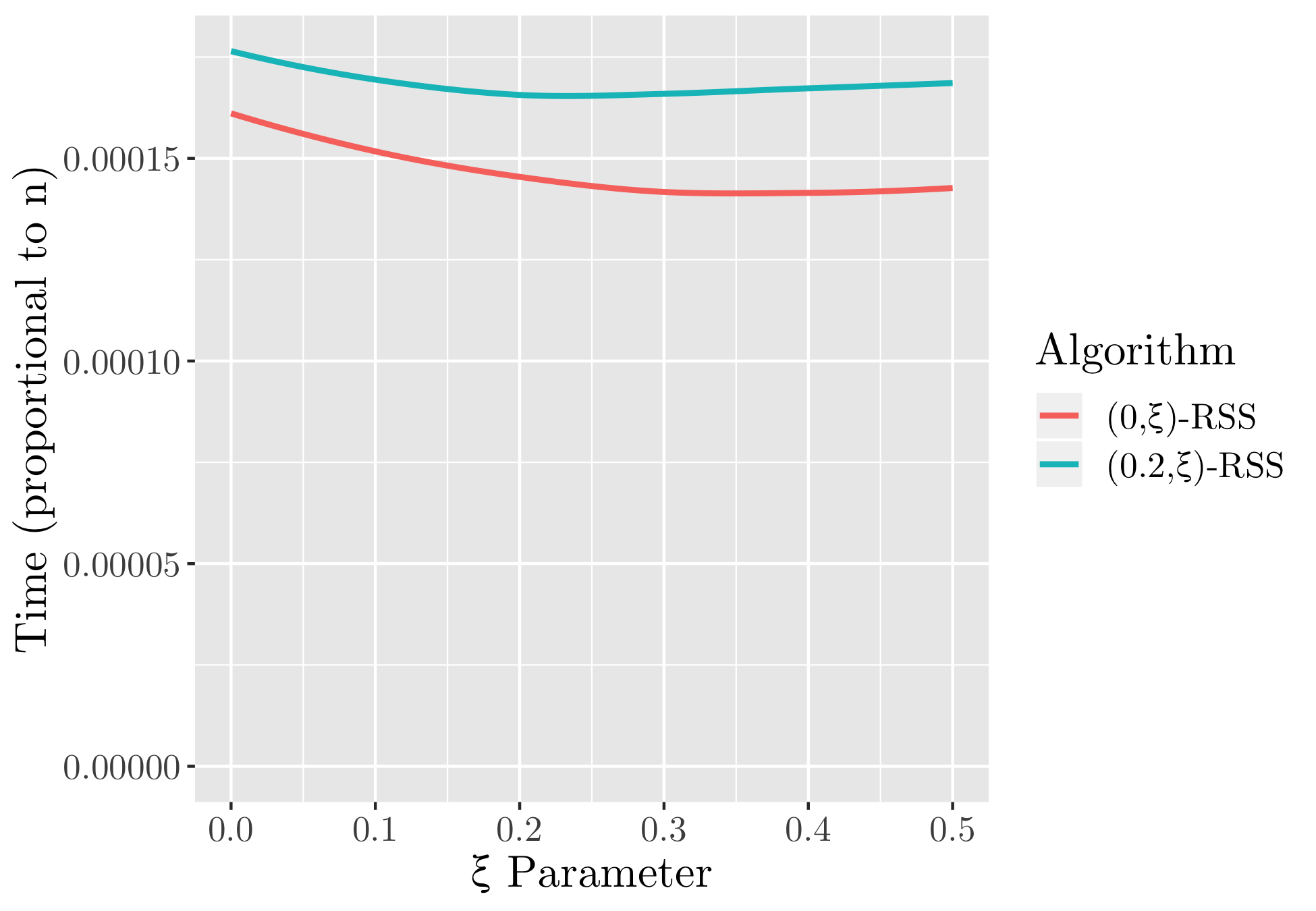}
        \vspace*{-10pt}
        \caption{Running time.}\label{fig:exp:eps:time}
    \end{subfigure}%
    \hfill
    \begin{subfigure}[b]{.48\linewidth}
        \centering\includegraphics[width=\textwidth]{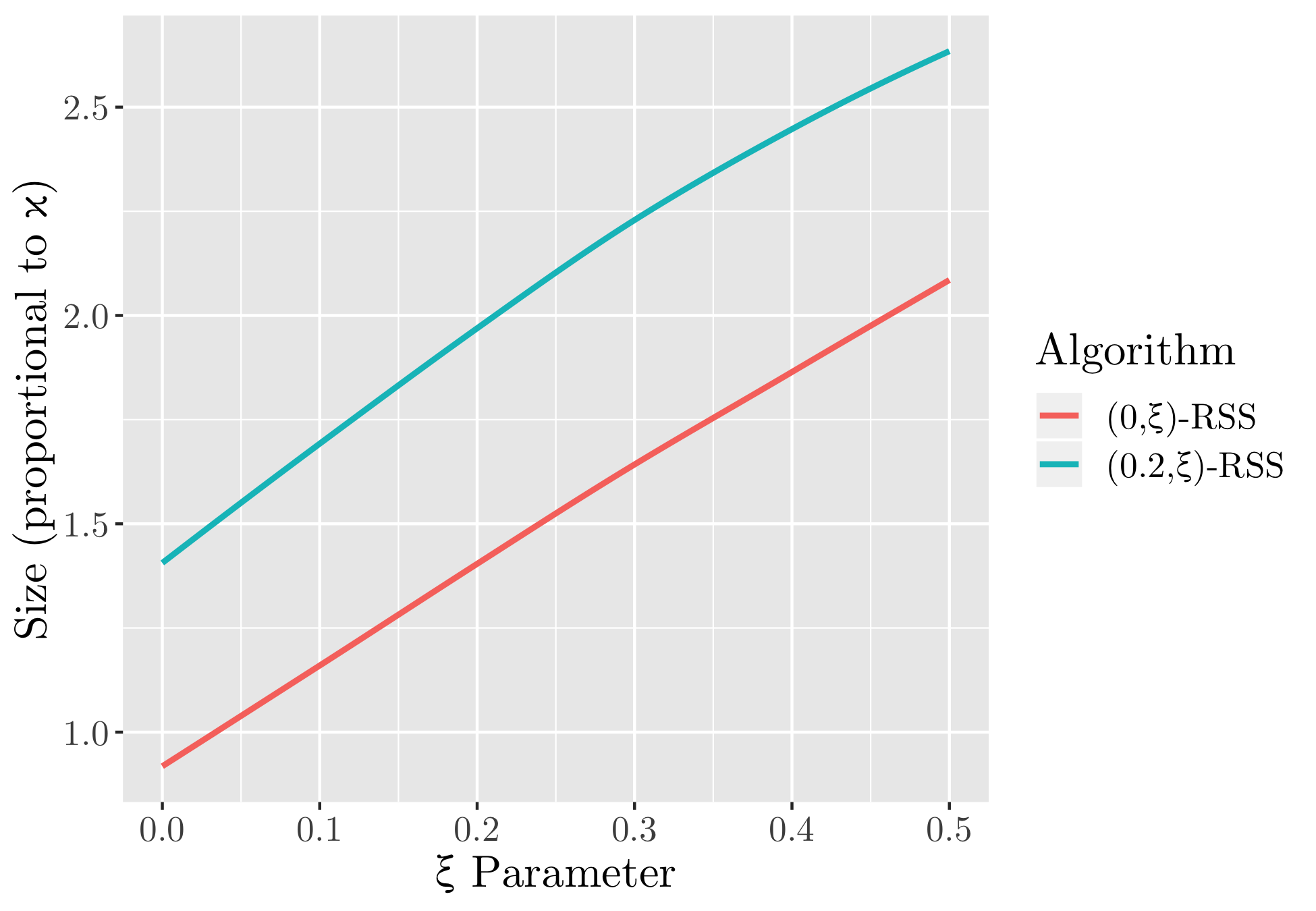}
        \vspace*{-10pt}
        \caption{Size of the selected subsets.}\label{fig:exp:eps:size}
    \end{subfigure}%
    \vspace*{-5pt}
    \caption{Evaluating the effect of increasing the parameter $\xi$ on \paramRSS{(\alpha,\xi)} for $\alpha=\{0, 0.2\}$.}\label{fig:exp:eps}
\end{figure*}

\begin{table}[h!]
\footnotesize
\centering
\begin{tabular}{||c|cccc||}
\hline\rule{0pt}{9pt}
Training set & $n$ & $d$ & $c$ & $\numNE\ (\%)$ \\ \hline
banana & 5300 & 2 & 2 & 811 (15.30\%) \\
cleveland & 297 & 13 & 5 & 125 (42.09\%) \\
glass & 214 & 9 & 6 & 87 (40.65\%) \\
iris & 150 & 4 & 3 & 20 (13.33\%) \\
iris2d & 150 & 2 & 3 & 13 (8.67\%) \\
letter & 20000 & 16 & 26 & 6100 (30.50\%) \\
magic & 19020 & 10 & 2 & 5191 (27.29\%) \\
monk & 432 & 6 & 2 & 300 (69.44\%) \\
optdigits & 5620 & 64 & 10 & 1245 (22.15\%) \\
pageblocks & 5472 & 10 & 5 & 429 (7.84\%) \\
penbased & 10992 & 16 & 10 & 1352 (12.30\%) \\
pima & 768 & 8 & 2 & 293 (38.15\%) \\
ring & 7400 & 20 & 2 & 2369 (32.01\%) \\
satimage & 6435 & 36 & 6 & 1167 (18.14\%) \\
segmentation & 2100 & 19 & 7 & 398 (18.95\%) \\
shuttle & 58000 & 9 & 7 & 920 (1.59\%) \\
thyroid & 7200 & 21 & 3 & 779 (10.82\%) \\
twonorm & 7400 & 20 & 2 & 1298 (17.54\%) \\
wdbc & 569 & 30 & 2 & 123 (21.62\%) \\
wine & 178 & 13 & 3 & 37 (20.79\%) \\
wisconsin & 683 & 9 & 2 & 35 (5.12\%) \\
v-100000-2-3-15 & 100000 & 2 & 3 & 1909 (1.90\%) \\
v-100000-2-3-5 & 100000 & 2 & 3 & 788 (0.78\%) \\
v-100000-3-3-15 & 100000 & 3 & 3 & 7043 (7.04\%) \\
v-100000-3-3-5 & 100000 & 3 & 3 & 3738 (3.73\%) \\
v-100000-4-3-15 & 100000 & 4 & 3 & 13027 (13.02\%) \\
v-100000-4-3-5 & 100000 & 4 & 3 & 10826 (10.82\%) \\
v-100000-5-3-15 & 100000 & 5 & 3 & 22255 (22.25\%) \\
v-100000-5-3-5 & 100000 & 5 & 3 & 17705 (17.70\%) \\
\hline
\end{tabular}
\vspace*{.3cm}
\caption{Training sets used to evaluate the performance of condensation algorithms. Indicates the number of points $n$, dimensions $d$, classes $c$, nearest-enemy points $\numNE$ (also in percentage \emph{w.r.t.} $n$).}\label{table:data}\vspace*{-.5cm}
\end{table}

\bibliographystyle{plainurl}
\bibliography{nnc}

\end{document}